\documentclass[11pt]{article}

\usepackage[colorlinks=true,linkcolor=blue,citecolor=blue]{hyperref}

\usepackage{amsmath, amssymb, amsthm}
\usepackage{mathtools}
\usepackage[capitalize,nameinlink]{cleveref}
\usepackage{fullpage}
\usepackage{graphics}
\usepackage{environ}
\usepackage{framed}
\usepackage{url}
\usepackage{algorithm}
\usepackage[noend]{algpseudocode}
\usepackage[labelfont=bf]{caption}
\usepackage{cite}
\usepackage{framed}
\usepackage[framemethod=tikz]{mdframed}
\usepackage{appendix}
\usepackage{graphicx}
\usepackage[textsize=tiny]{todonotes}
\usepackage{tcolorbox}

\newcommand\remove[1]{}

\newtheorem{theorem}{Theorem}
\newtheorem{lemma}{Lemma}[section]
\newtheorem*{lemma*}{Lemma}
\newtheorem{corollary}[lemma]{Corollary}
\newtheorem*{corollary*}{Corollary}

\newtheorem*{theorem*}{Theorem}
\newtheorem{definition}[lemma]{Definition}

\newtheorem*{rem*}{Remark}

\newcommand{\eps}{\varepsilon}
\renewcommand{\O}{\tilde{O}}

\newcommand{\RoundtripCover}{\textsc{RoundtripCover}}

\newcommand{\GoodCut}{\textsc{GoodCut}}
\newcommand{\bin}{B^{\mathrm{in}}}
\newcommand{\bout}{B^{\mathrm{out}}}
\newcommand{\lra}{\leftrightarrows}
\newcommand{\assign}{\leftarrow}
\newcommand{\iin}{i_{\mathrm{in}}}
\newcommand{\iout}{i_{\mathrm{out}}}
\newcommand{\BuildSimilar}{\textsc{BuildSimilar}}
\newcommand{\BallGrow}{\textsc{BallGrow}}
\newcommand{\RoundtripCoverNew}{\textsc{RoundtripCover2}}
\newcommand{\GoodCutNew}{\textsc{GoodCut2}}
\newcommand{\SimilarTest}{\textsc{Similar}}
\newcommand{\ON}{\mathrm{on}}
\renewcommand{\forall}{\text{ for all }}
\newcommand{\hn}{\hat{n}}

\newif\ifrandom
\randomtrue

\newcommand{\ceil}[1]{\lceil {#1} \rceil}

\newcommand{\otilde}{\tilde{O}}

\newcommand{\defeq}{:=}

\newcommand{\etal}{\textit{et~al.}}

\author{
Shiri Chechik\\
Tel Aviv University\\ 
{\tt shiri.chechik@gmail.com}
\thanks{
	This project has received funding from the European Union's Horizon 2020 research grant agreement 803118.
}
\and
Yang P. Liu \\
Stanford University \\
\texttt{yangpatil@gmail.com}
\thanks{
    Research supported by the U.S. Department of Defense via an NDSEG fellowship.
}
\and	
Omer Rotem \\
Tel Aviv University \\
{\tt omer.rotem1@gmail.com}
 \and
 Aaron Sidford\\
 Stanford University\\
\texttt{sidford@stanford.edu}
 \thanks{
 	Research supported by NSF CAREER Award CCF-1844855
 }
}

\begin{document}

\title{Constant Girth Approximation for Directed Graphs in \\ Subquadratic Time}

\begin{titlepage}
\clearpage\maketitle
\thispagestyle{empty}

\begin{abstract}

In this paper we provide a $\O(m\sqrt{n})$ time algorithm that computes a $3$-multiplicative approximation of the girth of a $n$-node $m$-edge directed graph with non-negative edge lengths. This is the first algorithm which approximates the girth of a directed graph up to a constant multiplicative factor faster than All-Pairs Shortest Paths (APSP) time, i.e. $O(mn)$. Additionally, for any integer $k \ge 1$, we provide a deterministic algorithm for a $O(k\log\log n)$-multiplicative approximation to the girth in directed graphs in $\O(m^{1+1/k})$ time. Combining the techniques from these two results gives us an algorithm for a $O(k\log k)$-multiplicative approximation to the girth in directed graphs in $\O(m^{1+1/k})$ time. Our results naturally also provide algorithms for improved constructions of roundtrip spanners, the analog of spanners in directed graphs.

The previous fastest algorithms for these problems either ran in All-Pairs Shortest Paths (APSP) time, i.e. $O(mn)$, or were due Pachocki~\etal~\cite{PRSTV18} which provided a randomized algorithm that for any integer $k \ge 1$ in time $\tilde{O}(m^{1+1/k})$ computed with high probability a  $O(k\log n)$ multiplicative approximation of the girth. Our first algorithm constitutes the first sub-APSP-time algorithm for approximating the girth to constant accuracy, our second removes the need for randomness and improves the approximation factor in Pachocki~\etal~\cite{PRSTV18}, and our third is the first time versus quality trade-off for obtaining constant approximations.
\end{abstract}

\end{titlepage}

\newpage

\section{Introduction}
The \emph{girth} of a graph $G$ is the length of the shortest cycle in $G$. It is an important graph quantity that has been studied extensively in both combinatorial settings (see Bollob\'{a}s's book \cite{Bollobas} for a discussion) and computational settings. In particular, exact algorithms for the girth running in time $O(mn)$ in weighted directed graphs \cite{OS17} are known. On the other hand, a result of Vassilevska W. and Williams show that a truly subcubic algorithm for girth (i.e. running in time $n^{3-\eps}$ for some $\eps > 0$) implies a truly subcubic algorithm for the All Pairs Shortest Path (APSP) problem \cite{VW10}. As it is a longstanding open problem whether APSP admits a truly subcubic time algorithm, exact computation of the girth in truly subcubic time would be a major breakthrough.

This has motivated the study of efficient \emph{approximation} algorithms for the girth. There has been extensive work on approximating the girth in \emph{undirected} graphs \cite{IR77, LL09, RW12, DKMS17}. Many such algorithms use the concept of a $\alpha$-\emph{spanner} of a graph $G$, a fundamental combinatorial object which was introduced by Chew \cite{Chew89}. An $\alpha$-spanner of a graph $G$ is a subgraph of $G$ which multiplicatively preserves distances up to a factor of $\alpha$. It is well-known that $(2k-1)$-spanners with $O(n^{1+ 1/k})$ edges exist for any undirected weighted graph \cite{ADDJ93}, and work on the efficient construction of such spanners \cite{TZ05, LTZ05, BS03} implies a $O(mn^\frac1k)$ time algorithm for $(2k-1)$-multiplicative girth approximation in undirected graphs. There has also been work on improved spanner constructions in the case of undirected unweighted graphs \cite{LL09, RW12}, and these algorithms also immediately imply algorithms for girth approximation in undirected unweighted graphs.

Therefore, in order to obtain efficient constant factor girth approximations in directed graphs, it is natural to study an analog of spanners in directed graphs. Unfortunately, approximately computing all pairs distances in directed graphs is a notoriously difficult problem and while sparse spanners do exist in all undirected graphs, they do not exist in all directed graphs. For example, any directed spanner for the ``directed" complete bipartite graph with $n$ vertices on the left directed towards $n$ vertices on the right clearly requires all $n^2$ edges. This problem seems to arise from the fact that the distance metric $d(u, v)$ in directed graphs is asymmetric. Therefore, if we want to construct sparse spanners, it is natural to work instead with the symmetric \emph{roundtrip distance} metric, defined as $d(u \lra v) := d(u, v) + d(v, u)$ \cite{CowenW04} and similarly define an $\alpha$-\emph{roundtrip spanner} of a directed graph $G$ to be a subgraph that multiplicatively preserves roundtrip distances up to a factor of $\alpha$.

Interestingly, there do exist roundtrip spanners for directed graphs with comparable sparsity as spanners for undirected graphs. A result of Roditty, Thorup, and Zwick \cite{RTZ08} shows that for any $k \ge 1$ and $\eps > 0$, every graph has a $(2k+\eps)$-roundtrip spanner with $O(k^2 n^{1+ 1/k}\log(nW) \eps^{-1})$ edges, where $W$ is the maximum edge weight. Unforunately, this algorithm ran in time $\Omega(mn)$, as it requires the computation of all pairs distances in the graph. Recent work Pachocki~\etal~\cite{PRSTV18} gave a randomized algorithm running in time $\tilde{O}(m^{1+1/k})$ which on weighted directed graphs $G$ returns a $O(k\log n)$-roundtrip spanner with $\tilde{O}(n^{1+1/k})$ edges and an $O(k\log n)$ approximation to the girth. Up to a logarithmic approximation factor, this matches the sparsity and runtime known for spanners on undirected weighted graphs and girth on sparse graphs.

The result of Pachocki~\etal~\cite{PRSTV18} constitutes one of small, but rapidly growing \cite{CohenKPPRSV17,CohenKKPPRS18},  set of instances where it is possible to obtain robust nearly linear time approximations to fundamental quantities of directed graphs in nearly linear time, overcoming typical running time gaps between solving problems on directed and undirected graphs.

However, a fundamental open problem left open by this work is whether it is possible to achieve subquadratic algorithms for constant factor approximation of the girth in directed graphs, and more ambitiously to fully close this gap and provide algorithms for $O(k)$ girth approximation and $O(k)$ roundtrip spanners in directed graphs that fully match the runtime and sparsity of those in undirected graphs. This is the primary problem this paper seeks to address and this paper provides multiple new girth approximation algorithms with improved runtime, approximation quality, and dependency on randomness.

\subsection{Our Results}

In this paper we provide a subquadratic algorithm for constant factor girth approximation in directed graphs and in turn show several improvements on the girth approximation algorithms and roundtrip spanner constructions in the work of Pachocki~\etal~\cite{PRSTV18}. Here and throughout the remainder of the paper we use $\tilde{O}(\cdot)$ notation to hide factors polylogarithmic in $n$, where $n$ is the number of vertices in the graph.

In \cref{sec:constantapprox} we consider obtaining constant approximations to the girth. In particular we provide a randomized algorithm that obtains a $3$-approximation to the girth on graphs with non-negative integer edge weights in $\O(m \sqrt{n})$ time. Up to logarithmic factors this matches the runtime that would be predicted from the fact that  $(2k - 1)$-undirected spanners with $\otilde(n^{1+1/k})$ edges can be constructed in $\otilde(mn^{1/k})$ time for $k = 2$. Further, we show that this procedure can be used to with high probability obtain constant multiplicative roundtrip spanners in directed graphs with arbitrary edge weights in $\otilde(m \sqrt{n})$ time.

\begin{theorem}[3-Multiplicative Girth Approximation]
\label{thm:3girth}
For any directed graph $G$ with $n$ vertices, $m$ edges, integer non-negative edge weights, and unknown girth $g$ we can compute in $\tilde{O}(m \sqrt{n})$ time an estimate $g'$ such that $g \le g' \le 3g$ with high probability in $n$.
\end{theorem}
\begin{theorem}[8-Multiplicative Roundtrip Spanners]
\label{thm:const_spanner}
For any directed graph $G$ with $n$ vertices, $m$ edges, integer non-negative edge weights, we can compute in $\tilde{O}(m\sqrt{n})$ time
an $8$-multiplicative roundtrip spanner with $\otilde(n^{3/2})$ edges with high probability in $n$.
\end{theorem}

Then, in \cref{sec:algo} we give algorithms for a $O(k\log\log n)$-multiplicative approximation of the the girth and construct $O(k \log\log n)$ multiplicative roundtrip spanners with $\O(n^{1+1/k})$ edges for a weighted directed graph $G$ with $n$ vertices and $m$ edges in $\tilde{O}(m^{1 + 1/k})$ time. These algorithms are deterministic and constitute the first deterministic nearly linear time algorithms for $\tilde{O}(1)$ multiplicative approximation of the girth and $\tilde{O}(1)$ multiplicative roundtrip spanners with $\tilde{O}(n)$ edges.

\begin{theorem}[Deterministic Multiplicative Girth Approximation]
\label{thm:girth}
For any integer $k \ge 1$ and weighted directed graph $G$ with $n$ vertices, $m$ edges, and unknown girth $g$ we can compute in $\tilde{O}(m^{1 + 1/k})$ time an estimate $g'$ such that $g \le g' \le O(k \log\log n) \cdot g$.
\end{theorem}
\begin{theorem}[Deterministic Multiplicative Roundtrip Spanners]
\label{thm:spanner}
For any integer $k \ge 1$ and any weighted directed graph $G$ with $n$ vertices and $m$ edges, we can compute in $\tilde{O}(m^{1 + 1/k})$ time
an $O(k \log\log n)$ multiplicative roundtrip spanner with $\tilde{O}(n^{1 + 1/k})$ edges.
\end{theorem}

Setting $k = \frac{\log n}{\log\log n}$ yields the following corollaries. For $k = \Omega(\log n)$ these results nearly match the optimal algorithms in undirected graphs for $O(k)$ girth approximation and the construction of $O(k)$ spanners.

\begin{corollary}
For any weighted directed graph $G$ with $n$ vertices, $m$ edges, and unknown girth $g$ we can compute in $\tilde{O}(m)$ time an estimate $g'$ such that $g \le g' \le O(\log n) \cdot g$.
\end{corollary}
\begin{corollary}
For any weighted directed graph $G$ with $n$ vertices and $m$ edges, we can compute in $\tilde{O}(m)$ time
an $O(\log n)$ multiplicative roundtrip spanner with $\tilde{O}(n)$ edges.
\end{corollary}

Interestingly, our results for constant factor randomized approximations and our results for deterministic approximations are achieved in different ways. Highlighting this, in \cref{sec:klogk_approx} we show how to combine the techniques of these algorithms to obtain both $O(k \log k)$ multiplicative approximations to the girth and $O(k \log k)$ multiplicative roundtrip spanners of size $\otilde(n^{1 + 1/k})$ in $\otilde(m n^{1/k})$ time with high probability in $n$.
\begin{theorem}[Constant Multiplicative Girth Approximation]
\label{thm:const_girth}
For any integer $k \ge 1$ and any weighted directed graph $G$ with $n$ vertices, $m$ edges, and unknown girth $g$ we can compute in $\tilde{O}(m^{1 + 1/k})$ time an estimate $g'$ such that $g \le g' \le O(k \log k) \cdot g$ with high probability in $n$.
\end{theorem}
\begin{theorem}[Constant Multiplicative Roundtrip Spanners]
\label{thm:arb_const_spanner}
For any integer $k \ge 1$ and any weighted directed graph $G$ with $n$ vertices and $m$ edges, we can compute in $\tilde{O}(m^{1 + 1/k})$ time an $O(k \log k)$ multiplicative roundtrip spanner with $\tilde{O}(n^{1 + 1/k})$ edges with high probability in $n$.
\end{theorem}
This shows that for any fixed $\eps > 0$, that there is an algorithm running in time $m^{1+\eps}$ that approximates the girth of a directed graph to within a constant depending on $\eps$, but not on $m$ or $n$. Additionally, this almost matches the $O(k)$-multiplicative girth approximation algorithms running in $m^{1+1/k}$ time in undirected graphs.

\subsection{Comparison to previous work} 

While the existence of roundtrip spanners matching the quality in undirected graphs was shown in \cite{RTZ08}, the runtime was $O(mn)$ and required an APSP computation. Our results, \cref{thm:3girth}, \cref{thm:const_spanner} are the first to show that constant factor girth approximation and construction of constant factor roundtrip spanners with $\O(n\sqrt{n})$ edges can be built in subquadratic $\O(m\sqrt{n})$ time. This algorithm leverages new randomized techniques for testing a notion we call \emph{similarity} between vertices not present in previous girth approximation and roundtrip spanner algorithms and we believe is of independent interest.

Our \cref{thm:girth} and \cref{thm:spanner} offer direct improvements over the analogous results in \cite{PRSTV18}. Specifically, our algorithms provide a tighter multiplicative girth approximation and multiplicative spanner stretch in the same runtime as the algorithms in \cite{PRSTV18}, which produce a $O(k \log n)$ girth approximation and $O(k\log n)$ roundtrip spanner with $\tilde{O}(n^{1 + 1/k})$ edges in time $\tilde{O}(m^{1 + 1/k}).$

Additionally, our algorithm is deterministic and in our opinion, simpler. The algorithm of Pachocki~\etal~\cite{PRSTV18} involved the following pieces. First, they use a method of Cohen to estimate ball sizes \cite{Cohen97} and resolve the case where there is a vertex whose inball and outball (of some small radius) intersect in a significant fraction of the vertices. In the other case, they use exponential clustering (see \cite{MPX13}) to partition the graph and recurse. Finally, they rerun the algorithm $n^{1/k}$ times. On the other hand, our algorithm simply grows inballs and outballs from various vertices, and uses a delicate cutting conditition to decide when to cut and recurse.

Additionally, \cref{thm:const_girth} and \cref{thm:arb_const_spanner} further improve upon Pachocki~\etal~\cite{PRSTV18} by combining the ideas from the constant factor girth approximation algorithm and the deterministic ball-growing algorithm, completely removing the dependence on $n$ in the approximation factor while still running in time $\O(m^{1+1/k})$. We remark that the ideas for our deterministic algorithm are essential in obtaining this last result, and that more directly combining the ideas of \cite{PRSTV18} with our constant factor approximation algorithm does not seem to give an $O(k\log k)$ multiplicative girth approximation in $\O(m^{1+1/k})$ time.

\subsection{Overview of Approach}
\label{sec:approach_overview}

\paragraph{Summary of randomized $O(1)$ approach.} Our approach to obtaining a $3$-approximation in \cref{sec:constantapprox} to the girth is rooted in the simple insight that if a vertex $v$ is in a cycle of length $R$ then every vertex in the ball of radius $\alpha$ from $v$ is at distance at most $\alpha + R$ from every vertex in the cycle. Consequently, for each vertex if we repeatedly prune vertices from its outball of radius $R$ if they do not have the property that they can reach every vertex in this ball by traversing a distance at most $2R$, then we will never prune away vertices in a cycle of length $R$ from that vertex. 

Leveraging these insights, we can show that if we randomly compute distances to and from a random $\otilde(\sqrt{n})$ vertices and if a cycle of length $O(R)$ is not discovered immediately then we can efficiently implement a pruning procedure so that each vertex only has in expectation $\otilde(\sqrt{n})$ vertices that could possibly be in a cycle of length $O(R)$ through that vertex. By then checking each of these sets for a cycle and being careful about the degrees of the vertices (and therefore the cost of the algorithm) this approach yields essentially a $4$-approximation to the girth in $\otilde(m\sqrt{n})$ time with high probability in $n$. 

Our $3$-approximation is then obtained by carefully applying this argument to both outballs and inballs and leveraging the simple fact that if a vertex $v$ is on a cycle  $C$ of length $R$ then for every $c \in C$ either $d(v,c) \leq R/2$ or $d(c,v) \leq R/2$.

\paragraph{Overview of deterministic $O(k \log \log n)$ results:} Our deterministic algorithm in \cref{sec:algo} is based on a different approach than our randomized constant approximation algorithms in \cref{sec:constantapprox}. We think this approach is of independent interest and further demonstrate its utility in \cref{sec:klogk_approx} by showing how to combine the insights that underly it with the algorithm from \cref{sec:algo} to achieve arbitrary constant approximations.

For the sake of simplicity, we focus on unweighted directed graphs $G$ and for a parameter $R$, construct a subgraph (roundtrip spanner) $H$ so that if the roundtrip distance between $u$ and $v$ is at most $R$ in $G$, then their roundtrip distance is at most $O(Rk\log\log n)$ in $H$. 

The key insight of guiding our algorithm is the following: instead of partitioning the graph into disjoint pieces and recursing (as is done in \cite{PRSTV18}), we instead allow the pieces to overlap on the boundaries. This is justified by the following observation. Consider a subgraph $W$ of $G$, and let $W'$ denote the subgraph consisting of all vertices within distance $R$ of $W$. Then if we recursively build a roundtrip spanner on $W'$, then we are guaranteed that we can delete $W$ from our graph. Indeed, if $u \in W$ and the roundtrip distance between $u$ and $v$ is at most $R$, then $u, v \in W'$. This simple observation allows us to overcome the critical challenge in \cite{PRSTV18}, arguing that that graph can be broken apart, while nevertheless preserving roundtrip distance.

This observation also forms the basis of an optimal spanner construction on unweighted undirected graphs, which appears in a book of Peleg (exercise 3 on page 188 in \cite{Peleg00}). Specifically, for any integer $k \ge 1$, we can construct a $(2k-1)$-spanner with $O(n^{1 + 1/k})$ edges in time $O(m).$ The construction works as follows. Start at any vertex $v$, let $B_i$ denote the ball of radius $i$ centered at $v$, and let $|B_i|$ denote the number of vertices in $B_i$. Grow such balls around $v$ until we find an index $i$ with $|B_{i+1}| \le n^{1/k} |B_i|.$ We can clearly guarantee that $i \le k$. At this point, add a spanning tree on $B_{i+1}$ to your spanner and delete all vertices in $B_i.$ Now, recurse on the remaining graph. It is easy to check that the resulting spanner is as desired. Our algorithm for directed graphs is similar, and we give a more specific overview in \cref{sec:4overview}.

We gain an $O(\log\log n)$ dependence over the undirected spanner algorithm presented because we must recurse on the balls we grew instead of simply building a spanning tree on them. The precise condition for recursion and corresponding calculation are performed in the algorithms \GoodCut~(\cref{algo:goodcut}) and \cref{lemma:inbound}.

Further, our $O(k\log k)$ approximations of \cref{sec:klogk_approx} are then achieved by using the techniques of the algorithms in \cref{sec:constantapprox} to better control the size of the outballs and inballs in an invocation of the deterministic algorithm of \cref{sec:algo}.
\section{Preliminaries}
\label{sec:prelim}

For weighted directed graph $G$, we let $V(G)$ and $E(G)$ denote the vertex and edge sets of $G$. We assume all edge lengths are nonnegative. For a subgraph $S \subseteq G$ (not necessarily vertex induced), let $V(S)$ denote the set of vertices of $G$, and let $E(S)$ denote the set of edges. For a subset $W \subseteq V(G)$, we define $G[W]$ to be subgraph induced by $W$. When the graph $G$ is clear from context, we let $n$ and $m$ denote $|V(G)|$ and $|E(G)|$ respectively. 

For a weighted directed graph $G$ with non-negative edge lengths, we let $d_G(u, v)$ denote the (shortest path) distance from $u$ to $v$ in $G$. When the graph $G$ is clear from context, we simply denote this as $d(u, v).$ If there is no path from $u$ to $v$, we let $d(u, v) = \infty.$ When $S$ is a subgraph of $G$, we let $d_S(u, v)$ denote the (shortest path) distance from $u$ to $v$ only using the edges in $E(S).$ We denote the \emph{roundtrip distance} between $u$ and $v$ as $d_G(u \lra v) := d_G(u, v) + d_G(v, u)$ and define a roundtrip spanner.

\begin{definition}[Roundtrip Spanner]
We say that a subgraph $S \subseteq G$ is an \emph{$\alpha$-roundtrip spanner} if $d_S(u \lra v) \le \alpha \cdot d_G(u \lra v)$ for all $u, v \in V(G).$
\end{definition}

For weighted directed graph $G$ we define the \emph{inball} and \emph{outball} of radius $r$ around a vertex $v$ as
 \[ \bin_v(r) \defeq G[\{ u : d(u, v) \le r \}] \text{ and } \bout_v(r) \defeq G[\{ u : d(v, u) \le r \}] \]
 respectively. In other words, the inball of radius $r$ around $v$ is the subgraph induced by vertices $u$ with $d(u, v) \le r.$ The outball is defined similarly. We define the \emph{ball} of radius $r$ around vertex $v$ as 
 \[ 
 B_v(r) \defeq G[\{ u : d(u \lra v) \le r \}]. 
 \] In other words, the ball of radius $r$ around $v$ is the subgraph induced by vertices $u$ within roundtrip distance $r$ of $v$.

\newcommand{\inset}{\mathrm{in}}
\newcommand{\outset}{\mathrm{out}}
\newcommand{\dir}{\mathrm{dir}}

\newcommand{\similarset}{\textsc{SimilarSet}}
\newcommand{\girthapprox}{\textsc{GirthApprox}}
\newcommand{\Grev}{G^{\mathrm{rev}}}

\newcommand{\similarsetspanner}{\textsc{SimilarSetSpanner}}

\section{Randomized Constant Approximations}
\label{sec:constantapprox}

Here we provide algorithms for efficiently computing a 3-approximation to the girth \cref{sec:girth3}. To simplify our algorithm and analysis we assume that the maximum degree of $G$ is bounded by $O(m/n)$, i.e. we assume it is only a constant larger than the average degree, which is $2m/n$.
We justify this assumption by showing that we can always reduce to this case as is formalized in the following lemma. We defer the proof to
\cref{proofs:regular}.
\begin{lemma}
\label{lemma:regular}
Given a directed weighted graph $G=(V,E)$ of $n$ vertices and $m$ edges with non negative edge weights, one
can construct a graph $H$ in $O(m)$ time of $O(n)$ vertices and $O(m)$ edges with non negative edge weights and of maximum degree
$O(m/n)$ such that
\begin{enumerate}
\item All roundtrip distances (between pairs of vertices in $G$) in $H$ and in $G$ are the same.
\item Given a cycle in $H$, one can find in $O(m)$ time a cycle in $G$ of the same length.
\item Given a subgraph $H'$ of $H$, one can find in $O(m)$ time a subgraph $G'$ of $G$ such that
the number of edges in $G'$ is at most the number of edges in $H'$ and the roundtrip distances in $H'$ and $G'$ are the same.
\end{enumerate}
\end{lemma}

\subsection{An $\tilde{O}(m \sqrt{n})$ Time $3$-approximation to Girth }
\label{sec:girth3}

In this section we show a procedure that given a directed weighted graph $G$ and a girth estimate $R$, returns
a cycle of length at most $3R$ if the girth in $G$ is at most $R$. The algorithm is given by $\girthapprox$ (See \cref{algo:girth3}) which in turn invokes the subroutine $\similarset$ (See \cref{alg:similar_vertices}).

In order to approximate the girth of $G$ we invoke this procedure for every $r=(1+\epsilon)^i$ for $1\leq i \leq
\log_{1+\epsilon}{n W}$ and stop once the procedure returns a cycle. If $g$ is the girth of $G$ this incurs an additional $\log_{1+\epsilon}{g}$ factor to the running time (as for the first index $i$ such that $(1+\epsilon)^i >g$ the algorithm will return a
cycle w.h.p.) and an additional $(1+\epsilon)$ factor in the approximation ratio.
The additional $(1+\epsilon)$ factor in the approximation ratio  can be avoided if the weights are integers by simply using binary search on the range between 1 and $nW$ (where $W$ is the maximum edge weight in $G$) and finding two consecutive integers $i$ and $i+1$ such that the procedure returned a cycle of length at most $3(i+1)$ when invoked on $i+1$ but not a cycle when invoked on $i$.
This incurs a $\log{nW}$ factor in the running time that can be improved to $O(\log{n})$ by the same method
as done in \cite{PRSTV18} of contracting small weight strongly connected components and deleting large weight edges (see Section 5.1 in \cite{PRSTV18} for more details).

Let $G =(V,E)$ be a directed graph with $n$ vertices and $m$ edges. 
We assume the graph $G$ is of average degree $\delta = 2m/n$ and that also the maximum degree in the graph is also $O(\delta)$.

The subroutine $\similarset$ gets as an input the graph $G$ and the target distance $R$ and either returns a cycle of length at most $3R$
or returns a subset $A_v$ of vertices for every $v\in V$. The subset $A_v$ for a vertex $v\in V$ consists of vertices at distance at most $R/2$ from $v$ with the guarantee that $A_v$ contains all vertices that are 
\begin{enumerate}
\item At distance at most $R/2$ from $v$ and
\item On a cycle of length $R$ with $v$.
\end{enumerate}
Procedure $\girthapprox$ invokes the Procedure $\similarset$ twice, once on $G$ and once on the reversed graph of $G$ (the graph obtained by reversing every edge of $G$). If a cycle of length $3R$ is returned in one of these calls then procedure $\girthapprox$ returns such a cycle.
Otherwise, let $\{A_v^\inset\}_{v \in V}$ be the sets returned from invoking $\similarset$ on the graph $G$ and
$\{A_v^\outset\}_{v \in V}$ on the reversed graph.
Next, the procedure for every $v\in V$ checks if there is a cycle containing $v$ of length at most $R$ in the induced graph
of $A_v^\inset \cup A_v^\outset$. If such a cycle exists then the procedure returns such a cycle.

Procedure $\similarset$ works as follows.
The algorithms starts by sampling $O(\log{n})$ independent subsets $S_i$ of expected size $O(\sqrt{n})$
each for $1 \le i \le M$ where $M=50\log{n}$. 
From every vertex $w \in \bigcup_{1 \le i \le M}{S_i}$ the algorithm runs Dijkstra from and to $w$ in $G$.
If a cycle of length $3R$ is detected then the algorithm returns it.

Next for every vertex $v\in V$ and index $1 \le i \le M$ the algorithm defines a set $T_i(v) \subseteq S_i$.
The sets $T_i(v)$ will be used to reduce the number of potential vertices that can be on a cycle of length at most $R$ with $v$.
First, the set $T_0(v)$ consists of all vertices in $S_0$ that are at distance at most $R/2$ from $v$.
Let $R_0(v)$ be a sampled set of $O(\log{n})$ vertices from $T_0(v)$.
Now, the sets $T_{i}(v)$ and $R_{i}(v)$ are defined as follows. The set $T_{i}(v)$ is the set of all vertices $s \in S_i$ such that $d(v,s) \le R/2$ and $d(s, t) \le 3R/2 \text{ for all } t\in \bigcup_{0 \le j \le i-1}{R_{j}(v)}$. Again, define $R_i(v)$ as a sampled set of $O(\log n)$ vertices from $T_i(v)$.

To gain intuition for the definition of $T_i(v)$ and $R_i(v)$, consider the set $G_i(v)$ of all vertices $s \in V(G)$ such that $d(v,s) \le R/2$ and $d(s, t) \le 3R/2 \text{ for all } t\in \bigcup_{0 \le j \le i-1}{R_{j}(v)}$. We remark that our algorithm does not compute $G_i(v)$, but its definition is essential for the analysis. Intuitively, the set $G_i(v)$ consists of the vertices after $i$ rounds that the algorithm still believes could be in a cycle of length $R$ with $v$. If $|G_i(v)| \ge 100\sqrt{n}\log n$, then by the choice of $S_i$ as an independent random set of expected size $O(\sqrt{n})$, we have that $T_i(v)$ is a random sample of $G_i(v)$ of expected size at least $100\log n.$ In this way, $R_i(v)$ is just a random sample of $G_i(v)$ of size $O(\log n).$

As we show in \cref{lem:SubsetShrinks}, if $|G_M(v)| \ge 100\sqrt{n}\log n$, our algorithm discovers w.h.p. a cycle of length at most $3R$ sometime during the shortest path computations done at the beginning. On the other hand, if $|G_M(v)| \le 100\sqrt{n}\log n$, then we can grow a shortest path tree from $v$ but only include vertices in $G_M(v)$ to search for a cycle of length $R$, only paying runtime $|G_M(v)| = \O(\sqrt{n})$ for that vertex $v$.

Formalizing this final step, the algorithm computes a shortest path tree $T(v)$ from $v$ up to depth $R/2$, keeping only vertices $s \in V$ such that $d(v,s) \le R/2$ and $d(s, t) \le 3R/2 \text{ for all } t\in \bigcup_{0 \le j \le M}{R_{j}(v)}$.
The set $A_v$ is the set of vertices in $T(v)$.

\begin{algorithm}[h!]
	\caption{$\girthapprox(G,R)$,  takes a graph $G$ and a parameter $R$.  If the girth of $G$ is at most $R$ this algorithm outputs w.h.p. a cycle of length at most $3R$}
	\label{algo:girth3}
	\begin{algorithmic}[1]
		\State Invoke $\similarset(G,R)$ to either find a cycle of length at most $3R$ or set $A_v^\outset \subseteq V$ for each $v \in V(G)$. \label{line:similar_search_1}
		\State Invoke $\similarset(\Grev,R)$ where $\Grev$ is the graph where the direction of every edge is reversed to either find a cycle of length at most $3R$ or sets $A_v^\inset \subseteq V$ for each $v \in V(G)$.
		\State If a cycle of length at most $3R$ has yet to be found for each $v \in V(G)$ perform Dijsktra from $v$ in the graph induced by $A_v^{\outset} \cup A_v^{\inset}$ to find a cycle of length at most $R$ through one of the $v$.
		\label{line:search_union}
		\State Return any cycle of length at most $3R$ found. \label{line:final_ball-search}
	\end{algorithmic}
\end{algorithm}

\begin{algorithm}[h!]
\caption{$\similarset(G, R)$, takes a graph $G$ and a parameter $R$. This algorithm either computes a cycle of length $3R$ or a set $A_v \subseteq V(G)$ of ``similar''
	vertices to $v$ (with respect to balls of radius $R/2$) for each $v \in V(G)$.}
\label{alg:similar_vertices}
\begin{algorithmic}[1]
\State For $M = 50\log n$, sample sets $S_0, S_1, \cdots, S_M \subseteq V(G)$, each of expected size $O(n^{1/2})$ by sampling every vertex
$v \in V$  independently with probability $p = n^{-1/2}$.
\State Run Dijkstra to/from each vertex $v\in S_i$ for every $1 \leq i \leq M$.
\State If there exists a vertex $v \in \cup_{1 \le i \le M}{S_i}$ such that $v$ is on a cycle of length $3R$ then return the shortest such cycle.
    \label{line:returncycledetected}
\For{every vertex $v \in V$}
	\State Set $T_0(v) \gets \{s \in S_0 \mid d(v, s) \le R/2 \}$.
	\For{$i = 1 , \dots , M$}
		\If{$|T_{i-1}(v)| > 100 \log{n}$}
		\State Let $R_{i-1}(v)$ be $100 \log{n}$ vertices chosen independently at random from $T_{i-1}(v)$
		\Else
		\State Let $R_{i-1}(v) = T_{i-1}(v)$.
		\EndIf
		\State
		$T_i(v) \gets \{ s \in S_{i} \mid d(v, s) \le R/2  \text{ and }
		d(s, t) \le 3R/2 \text{ for all } t\in \cup_{0 \le j \le i-1}{R_{j}(v)}\}$
   \EndFor
	\State  Compute a shortest path tree $T(v)$ up to depth $R/2$ keeping only vertices $s$ such that
            $d(s, t) \le 3R/2 \text{ for all } t\in \cup_{0 \le j \le M}{R_{j}(v)}\}$.
    \State   Set $A_v$ to be the set of vertices in $T(v)$.
    \label{line:GrowSmallBall}
	\EndFor
	\Return $A_v$ for all $v \in V$
\end{algorithmic}
\end{algorithm}

Next we  prove the correctness of our girth computation algorithm  $\girthapprox$ (\cref{algo:girth3}) and bound its running time. First we prove the following lemma which provides a fairly straightforward argument that the algorithm always outputs the correct result. The more challenging part of the analysis will be to bound its running time.

\begin{lemma}
\label{lem:correctness}
If $G$ contains a cycle of length at most $R$ then $\girthapprox(G,R)$ (\cref{algo:girth3}) returns a cycle of length at most $3R$.
\end{lemma}
\begin{proof}
Assume $G$ contains a cycle $C$ of length at most $R$.
Let $v$ be a vertex in $C$.
If the algorithm returns a cycle in line \ref{line:returncycledetected} of $\similarset$ (\cref{alg:similar_vertices}) then since this cycle has length at most $3R$, the algorithm works as desired.

Consequently, we assume that this is not the case. Our goal is now to show that $A_v^\outset$ contains all vertices $c \in C$ such that $d(v,c)\leq R/2$ and that $A_v^\inset$ contains all vertices $c \in C$ such that $d(c,v)\leq R/2$. Since for all $c \in C$ either $d(v,c)\leq R/2$ or $d(c,v)\leq R/2$ this will imply that $C \subseteq A_v^\outset \cup A_v^\inset$ and therefore a cycle of length at most $R$ will be found in Line~\ref{line:search_union} of $\girthapprox(G,R)$ (\cref{algo:girth3}) and the algorithm works as desired. Further, note that it suffices to show that  $A_v^\outset$ contains all vertices $c \in C$ such that $d(v,c)\leq R/2$ as this will imply the desired claimed regarding $A_v^\inset$ by symmetry.

Consider the execution of $\similarset$ (\cref{alg:similar_vertices}) from Line~\ref{line:similar_search_1} of $\girthapprox(G,R)$ (\cref{algo:girth3}). Further, consider a vertex $t \in T_i(v)$ for some $0 \le i \le M-1$.
Recall that $d(v,t) \leq R/2$ (by definition and construction of $T_i(v)$).
Consider a vertex $c$ in $C$. As $v$ and $c$ are on a cycle of length $R$ we have $d(c,v) \leq R$ and therefore
$d(c,t) \leq d(c,v) + d(v,t) \leq 3R/2$ by triangle inequality.
It follows by construction that each vertex $c \in C$ with $d(v,c)\leq R/2$ will be added to $A_v$ as desired.
\end{proof}

With the correctness of $\girthapprox$ (\cref{algo:girth3}) established, in the remainder of this section we focus on analyzing its running time. To do this we will consider an invocation of $\similarset$ (\cref{alg:similar_vertices}) and both bound its running time and the size of the sets $A_v$ it computes.

Before setting up the proofs, for each vertex $v \in V$ we define
\[
G_0(v) = \{s \in V \mid d(v, s) \le R/2 \}
\]
and
\[
G_i(v) = \{s \in V \mid d(v, s) \le R/2 \text{ and }
d(s, t) \le 3R/2 \text{ for all } t\in \cup_{0 \le j \le i-1}{R_{j}(v)}\} ~.
\]
Notice that the distribution of
$T_i(v)$ is the distribution on vertices that results from taking each $s \in G_i(v)$ and including it in $T_i(v)$ with probability $p = 1/\sqrt{n}$.

Loosely speaking, the analysis of the running time is roughly as follows.
The main non trivial part is to show that the expected size of the sets $A_v^\inset$ and $A_v^\outset$ is $\tilde{O}(\sqrt{n})$.
This, together with the assumption that the maximum degree is $O(m/n)$, will imply that the running time of our algorithm is $\tilde{O}(m \sqrt{n})$.
We roughly speaking show the following for the set $A_v^\outset$ (similarly for the set $A_v^\inset$ ).
We want to claim that w.h.p. the sets $G_i(v)$ are decreasing by at least a constant factor until there is a set $G_i(v)$ of $\tilde{O}(\sqrt{n})$ size.
As $A_v^\outset$ is a subset of $G_M(v) \subseteq G_{i}(v)$, the claim follows.
Assume this is not case, i.e., there exists an index $i$ such that $|G_{i+1}(v)| > 0.8 |G_{i}(v)|$.
Note that by construction for every vertex $s$ in $G_{i+1}(v)$
all vertices in $R_{i}(v)$ are at distance at most $3R/2$ from it.
As $R_{i}(v)$ is a sampled set of $G_{i}(v)$, we can show that w.h.p. most vertices in $G_i(v)$ (say 0.9 fraction of them) are at distance at most $3R/2$ from $s$.
As $|G_{i+1}(v)| > 0.8 |G_{i}(v)|$, this means that this is also true for most vertices in $G_{i}(v)$.
That is, most vertices in $G_{i}(v)$ are at distance at most $3R/2$ to most of the other vertices in $G_{i}(v)$.
We show by counting argument that in this case there must be many pairs of vertices $u$ and $v$ such that $u,v \in G_i(v)$
and $d_G(u,v) \leq 3R/2$ and $d_G(v,u) \leq 3R/2$ (hence $u$ and $v$ are both on a cycle of length at most $3R$).
That is, w.h.p. $G_i(v)$ contains many vertices that are on cycles of length at most $3R$.
W.h.p. we can show that such a vertex will belong to $S_i(v)$ and therefore the algorithm will detect a cycle of length $3R$ and will not continue to computing the sets
$A_v^\outset$.

\begin{lemma}
\label{lem:ManyEdges}
Consider a vertex $v$, index  $i \in [M]$ such that $|G_{i}(v)| \geq 200 \sqrt{n} \log{n}$ and a vertex $u \in V$.
If there are less than $0.9|G_i(v)|$ vertices $s \in G_i(v)$ such that $d(u,s) \leq 3R/2$ then
with probability at least $1-2/n^{10}$, $u\notin G_{i+1}(v)$.
\end{lemma}
\begin{proof}
Note that the distribution of obtaining $T_j(v)$  is equivalent to the distribution
of picking every vertex in $G_j(v)$ with probability $p$ for every $1 \le j \le M$.

We first show that with high probability $T_{i}(v)$ contains at least $100 \log{n}$ vertices (and therefore also $R_{i}(v)$).
As $|G_{i}(v)| \geq 200 \sqrt{n} \log{n}$ then
the expected size of $|T_{i}(v)|$ is at least $200 \log{n}$.
Therefore, by Chernoff Bound the probability that $|T_{i}(v)| \le 100 \log{n}$ is at most
$\left(\frac{e^{-1/2}}{1/2^{1/2}}^{100 \log{n}}\right) < 1/n^{10}$.

Assume this is indeed the case, that is, $T_{i}(v)$ contains at least $100 \log{n}$ vertices.
The set $R_{i}(v)$ is a sampled set of $100 \log{n}$ vertices from $T_{i}(v)$.
As the distribution of obtaining the set $T_{i}(v)$ is equivalent to distribution of picking every vertex in
 $G_{i}(v)$ with probability $p$ then the distribution of $R_{i}(v)$ is
 equivalent to picking $100 \log{n}$ vertices from $G_{i}(v)$ (every vertex in $G_{i}(v)$ has the same probability appearing in $R_{i}(v)$).
Consider a uniformly random vertex $s$ from $G_{i}(v)$.
With probability at least $1/10$ we have $d(u,s) > 3R/2$.
In other words with probability at most $9/10$ we have $d(u,s) \leq 3R/2$.
Therefore, the probability that for every vertex $s$ in $R_j(v)$ we have  $d(u,s) \leq 3R/2$
is at most $(9/10)^{100 \log{n}}\leq 1/n^{10}$.

The lemma follows by union bound over the events that either $|T_{i}(v)|$ is smaller than $100 \log{n}$ or for all $s \in R_i(v)$ we have $d(u,s) \leq 3R/2$.
\end{proof}

\begin{lemma}
\label{lem:SubsetShrinks}
If there exists a vertex $v$ and
an index $i$ such that
$|G_{i}(v)| \geq 200 \sqrt{n} \log{n}$ and $|G_{i+1}(v)| \geq 0.8 |G_{i}(v)|$
then with probability
at least $1-1/n^{8}$
there exists a vertex in $T_i(v)$ that is contained in a cycle of length at most $3R$.
\end{lemma}
\begin{proof}
Assume such a vertex $v$ and index $i$ exist.

We say that a vertex $u$ is $(v,i)$-dense if
there are at least $0.9|G_i(v)|$ vertices $s \in G_i(v)$ such that $d(u,s) \leq 3R/2$.

By union bound on all vertices $v\in V$
on \cref{lem:ManyEdges},
with probability at least $1-2/n^{9}$,
all vertices in $G_{i+1}(v)$  are $(v,i)$-dense.

As $G_{i+1}(v) \subseteq G_{i}(v)$ and $|G_{i+1}(v)| \geq 0.8|G_{i}(v)|$, we also have that  with probability at least $1-2/n^{9}$,
$0.8|G_{i}(v)|$ vertices in $G_{i}(v)$  are $(v,i)$-dense.
Assume this is indeed the case.

Imagine constructing the following directed graph $H$ whose set of vertices is $G_{i}(v)$ and set of edges is the following.
For every vertex  $u$ in $G_{i}(v)$ add an outgoing edge for every vertex $s$ such that $d(u,s) \leq 3R/2$.
Note that if there exists two edges in the graph $(u,s)$ and $(s,u)$ then both $u$ and $s$ are on a cycle of length at most $3R$.
We next show that by counting argument there are many vertices in $G_{i}(v)$ that are on a cycle of length at most $3R$.
Every $(v,i)$-dense vertex $u$ has $0.9|G_{i}(v)|$ outgoing edges in $H$.
There are at least $0.9|G_{i}(v)|$ $(v,i)$-dense vertices in $H$.
We get that the number of edges $E(H)$ is at least $0.71|G_{i}(v)|^2$, that is, $|E(H)| \geq 0.71|G_{i}(v)|^2$.

On the other hand let $\alpha$ be the fraction of vertices in $G_{i}(v)$ that do appear on a cycle of length at most $3R$.
For every edge in $H$ give a  credit of 1/2 for each of its endpoints vertices.
Note that every vertex $x$ that do not belong to a cycle of length at most $3R$ can get a credit of less than $|G_{i}(v)|/2$.
To see this, note that there is no other vertex with both incoming and outgoing edge to $x$ (as otherwise $x$ is on a cycle of length at most $3R$) so the total number of incoming and outgoing edges of $x$ is at most $|G_{i}(v)|-1 < |G_{i}(v)|$.
Hence, the total credit of $x$ is less than $|G_{i}(v)|/2$.
The total credit of a vertex $x$ that do participate in a cycle of length at most $3R$ is less than $|G_{i}(v)|$.
We get that the total credit of all vertices, which is also equal to the total number of edges in $H$, is
less than $\alpha |G_{i}(v)| |G_{i}(v)|/2 +  (1-\alpha) |G_{i}(v)|^2$.
It follows that
$0.71|G_{i}(v)|^2 \leq \alpha |G_{i}(v)| |G_{i}(v)|/2 +  (1-\alpha) |G_{i}(v)|^2$.
Straight forward calculation show that $\alpha < 0.58$ and thus $1-\alpha > 0.42$.
In other words, at least $0.42|G_{i}(v)|$ vertices in $G_{i}(v)$  belong to a cycle of length at most $3R$.

Next, we claim that w.h.p. there is such a vertex in $T_i(v)$.
Recall that the distribution of $T_i(v)$ is equivalent to picking every vertex in $G_i(v)$ with probability $p$.
Consider one vertex that participates in a cycle of length at most $3R$ the probability it does not belong to
$T_i(v)$ is $1-p$.
The probability that none of the $0.42|G_{i}(v)|$ vertices belong to $T_i(v)$ is at most
$(1-p)^{0.42|G_{i}(v)|} \leq (1-p)^{84 \log{n}/p} \leq 1/n^{10}$.

The lemma follows (as $ 1/n^{10} + 2/n^{9} < 1/n^{8}$ for large enough $n$).
\end{proof}

Finally, the following concludes the running time of our algorithm.

\begin{lemma}
\label{lem:totaltime}
The expected running time of \cref{algo:girth3} is $O(m\sqrt{n}\log{n}+n\sqrt{n}\log^3{n}) = \tilde{O}(m \sqrt{n})$.
\end{lemma}
\begin{proof}
Consider one of the executions of $\similarset$ (\cref{alg:similar_vertices}) by $\girthapprox$ (\cref{algo:girth3}). This algorithm computes Dijkstra to/from each vertex $w\in S_i$ for every $1 \leq i \leq M$ in $O(m+n\log{n})$ time.
The expected size of each $S_i$ is $O(n^{1/2})$. Thus, the expected time of this computation for $S_i$ is
$O(m\sqrt{n}+n\sqrt{n}\log{n})$. There are $O(\log{n})$ sets $S_i$ and therefore there is at most
$O(m\sqrt{n}\log{n}+n\sqrt{n}\log^2{n})$ expected time for the computation of all Dijkstra's.
Next, for every vertex $v$ the algorithm computes the sets $T_i(v)$ for every $i\in [M]$.
The set $T_0(v)$ can be computed easily in $O(|S_0|)$ time which is $O(n^{1/2})$ in expectation.
In order to compute $T_i(v)$ for $i >0$, the
algorithm considers every vertex $s \in S_i$ and it check if $s$ is at distance at most $3R/2$ from
every vertex in $t\in \cup_{j\in[0,...,i-1]}{R_{j}(v)}$.
There are $O(\log^2{n})$ vertices $t$ in $\cup_{j\in[0,...,i-1]}{R_{j}(v)}$.
The distance $d(s,v)$ is already computed and thus can be retrieved in $O(1)$ time.
Overall, computing the set $T_i(v)$ takes $O(n^{1/2} \log^2{n})$ in expectation.
Therefore, $O(n^{1/2} \log^3{n})$ for all indices $i \in [M]$.
Hence, for all vertices $v$ $O(n^{3/2} \log^3{n})$ expected time for this part.

Next, we bound the cost of computing the balls, $A_v$, and we bound their size. By a slight abuse of notation we call a vertex $s$ $(v,M)$-dense if it satisfies \[ d(v, s) \le R/2 \text{ and }
d(s, t) \le 3R/2 \text{ for all } t\in \cup_{j\in[0,...,M]}{R_{j}(v)}.\]
The algorithm grows a ball from every vertex $v\in V$ by only keeping vertices $s$
that are $(v,M)$-dense to compute $A_v$.

We first show that if there is no index $i$ such that $|G_{i}(v)| \geq 200 \sqrt{n} \log{n}$ and
$|G_{i+1}(v)| \geq 0.8 |G_{k}(v)|$
then the expected time to compute the ball of $v$ is $O(n^{1/2}\log^3{n} +  n^{1/2} \log{n} \cdot \delta)$.
We do that by showing that the expected number of vertices
in $G_M(v)$ is $O(n^{1/2} \log{n})$.
As the maximum degree in $G$ is $O(\delta)$
and checking if a vertex $s$ is $(v,M)$-dense takes $O(\log^2{n})$ time, then the claim follows.

As for every $i$ such that $|G_{i}(v)| \geq O(\sqrt{n}\log{n})$ we have
$|G_{i+1}(v)| \leq 0.8 |G_{i}(v)|$
then straight forward calculation shows that there exists an index $M' \in[1..M]$ such that
$|G_{M'}(v)| < O(\sqrt{n}\log{n})$.
Note that the ball of $v$ contains only vertices from $G_{M'}(v)$ and thus the claim follows.

We now assume that there exists a vertex $v$ and index $i$ such that $|G_{i}(v)| \geq O(\sqrt{n} \log{n})$ and
$|G_{i+1}(v)| \geq 0.8 |G_{i}(v)|$.
By claim \ref{lem:SubsetShrinks} in this case with probability at least  $1-1/n^8$ the algorithm finds a cycle of length $3R$ and returns it in Line \ref{line:returncycledetected}. Therefore, in this case the algorithm does not compute the balls in Line \ref{line:GrowSmallBall}.
With probability at most $1/n^8$ the algorithm does not find a cycle in Line \ref{line:returncycledetected} and therefore continues to computing the balls in Line \ref{line:GrowSmallBall}. The computation of all balls in Line \ref{line:GrowSmallBall} is bounded by $O(mn)$ in this case.
As this happens with very small probability this does not effect the asymptotic bound of the expected running time.
The lemma follows.
\end{proof}

We conclude this section with the proof of \cref{thm:3girth}.

\begin{proof}[Proof of \cref{thm:3girth}]
The algorithm calls Algorithm $\girthapprox$ using a binary search
on the range $[1, nW]$ to find a parameter $R$ such that Algorithm $\girthapprox$ returns a cycle (of length at most $3(R+1)$) when invoked on $R+1$
but not on $R$.
As mentioned above the dependency on $\log{nW}$ can be improved to $\log{n}$
using the method used in \cite{PRSTV18} (Section 5.1).
Roughly speaking this method constructs in $O(m \log{n})$ time a set of graphs such that the number of vertices in all these graphs together is $O(n \log{n})$, the
number of edges is $O(m \log{n})$, the
ratio between the maximum edge  weight and the minimum edge weight in all these graphs is $O(n)$ and the shortest cycle is contained in one of these graphs. Instead of running binary search on $G$, we run it in each of these graphs.

Now using  \cref{lem:correctness} and \cref{lem:totaltime} the theorem follows.
\end{proof}
We give our result on constant approximation roundtrip spanners in $\O(m\sqrt{n})$ time and show \cref{thm:const_spanner} in \cref{sec:constspanner}.

\section{Deterministic $O(k \log \log n)$ Approximation Algorithms}
\label{sec:algo}
In this section we present our deterministic algorithms for computing a $O(k \log \log n)$ approximation to the girth and computing $O(k \log \log n)$ multiplicative roundtrip spanners. Our main result will be showing how to compute improved roundtrip covers as defined originally in \cite{RTZ08}. Leveraging this result we will prove \cref{thm:girth} and \cref{thm:spanner}.

First, leveraging the definitions of balls in \cref{sec:prelim} we define roundtrip covers.  Intuitively, roundtrip covers are a union of balls of radius $kR$ such that if vertices $u, v \in V(G)$ satisfy $d(u \lra v) \le R$ then $u, v$ are both in some ball in the cover.
\begin{definition}[Roundtrip Covers]
A collection $C$ of balls is a \emph{$(k, R)$ roundtrip cover} of a weighted directed graph $G$ if and only if every ball in $C$ has radius at most $kR$, and for any $u, v \in V(G)$ with $d(u \lra v) \le R$ there is a ball $B \in C$ such that $u, v \in B$.
\end{definition}
Specifically, we show the following theorem.
\begin{theorem}[Improved Roundtrip Covers]
\label{thm:cover}
For an $n$-vertex $m$-edge graph $G$, an execution of $\RoundtripCover(G, k, R)$ returns a collection $C$ of balls that forms a $(O(k\log\log n), R)$ roundtrip cover of a weighted directed graph $G$ in time $m^{1+O(1/k)}$ where $\sum_{B \in C} |V(B)| = n^{1+O(1/k)}.$
\end{theorem}
To show \cref{thm:girth} from \cref{thm:cover}, we can compute $(k,2^i)$ roundtrip covers for all $0 \le i \le O(\log n)$, and set our girth estimate as the minimum radius of any ball in the cover that has a cycle. To compute a roundtrip spanner, simply take the union of all the balls in the $(k,2^i)$ roundtrip covers for all $i = O(\log n).$

The rest of the section is organized as follows. In \cref{sec:mainalgo} we state our main algorithm. In \cref{sec:analysis} we analyze the algorithm and prove \cref{thm:cover}. In \cref{sec:mainappl} we use \cref{thm:cover} to formally prove \cref{thm:girth} and \cref{thm:spanner}.
\subsection{Technical Overview}
\label{sec:4overview}
We focus on unweighted directed graphs $G$ and for a parameter $R$, construct a roundtrip spanner $H$ so that if the roundtrip distance between $u$ and $v$ is at most $R$ in $G$, then their roundtrip distance is at most $O(Rk\log\log n)$ in $H$.

Our approach is based on growing inballs and outballs in the graph $G$. Fix a vertex $v$, and let $\bin_i, \bout_i$ denote the inball and outball of radius $iR$ around $v$, and let $|\bin_i|, |\bout_i|$ denote the number of vertices in the balls and fix $d = O(k \log\log n).$ We start by growing and inball and outball around $v$. First, if $|\bin_d \cap \bout_d| \ge \frac{n}{2}$, then we can build a roundtrip ball of radius $2dR+R$ and delete $\bin_d \cap \bout_d$ from our graph. This is safe essentially by our observation above. Otherwise, we find an index $i$ such that $|\bin_{i+1}|$ isn't much larger than $|\bin_i|$, we recursively build a roundtrip cover on $\bin_{i+1}$ and then delete $\bin_i$. This is safe to do by our observation above. Similarly, if there is an index $i$ such that $|\bout_{i+1}|$ isn't much larger than $|\bout_i|$, we recursively build a roundtrip cover on $\bout_{i+1}$ and then delete $\bout_i$. Through standard ball cutting inequalities we can show that such an index $i$ exists (\cref{lemma:inbound}). We would like to elaborate on a few points. First, when we compare the sizes of $|\bin_{i+1}|$ and $|\bin_i|$, we compare both the number of vertices \emph{and} edges, the former to control the size of the roundtrip spanner constructed, and the latter to control runtime. Second, we grow the inball and outball at the same rate, i.e. we alternately add an edge at a time to the inball and outball to maintain that the work spent on each is the same.

\subsection{Main Algorithm}
\label{sec:mainalgo}
We first give a high-level description of our algorithm for computing Roundtrip Covers, \RoundtripCover, which is presented formally as \cref{algo:roundtripcover}.
\paragraph{High-level Description of Algorithm.} As discussed in \cref{sec:approach_overview,sec:4overview}, our algorithm is based on ball growing along with the following observation: if for a radius $r'$
we compute a roundtrip cover of $\bin_v(r'+R)$ and add all the balls in the computed roundtrip cover on $\bin_v(r'+R)$ to our final cover, then we can safely delete all vertices $u \in \bin_v(r')$ from our graph and recurse on the rest of graph; the deleted vertices are already satisfied in the sense that for every $u' \in V(G)$ with $d(u \lra u') \le R$ there is a ball $B$ in the cover such that $u, u' \in B$.
Indeed, if $u \in \bin_v(r')$ and $d(u \lra u') \le R$ then $u, u' \in \bin_v(r'+R)$ and therefore we are guaranteed that the roundtrip cover on $\bin_v(r'+R)$ contains a ball $B$ such that $u,u' \in B$.
Using this observation, we grow inballs and outballs around vertices in our graph $G$ to ``partition" our graph into pieces that possibly overlap, where the overlap corresponds to the boundary $\bin_v(r'+R) \backslash \bin_v(r')$ in our example.

We describe our algorithm in more detail now. Consider any vertex $v$. We grow an inball and outball around $v$ at the same rate, spending the same time on the inball and outball. First, we consider the case that $|V(\bin_v(r))|, |V(\bout_v(r))| \ge \frac{3n}{4}$ for some $r = O(Rk \log\log n)$, as was done in Pachocki~\etal~\cite{PRSTV18}. Then we know that $|V(\bin_v(r)) \cap V(\bout_v(r))| \ge \frac{n}{2}.$ By our observation above, we can add the ball $B_v(2r+R)$ to our roundtrip cover, delete $\bin_v(r) \cap \bout_v(r)$ from $G$, and recurse on the remainder. Otherwise, if we find a radius $r'$ such that say $\bin_v(r')$ and $\bin_v(r'+R)$ satisfy the conditions of \GoodCut~(\cref{algo:goodcut}), then we recurse on $\bin_v(r'+R)$ and delete $\bin_v(r')$ from our graph and recurse on the remaining graph. This is safe to do by our observation above. We can also do an analogous process on $\bout_v(r')$ and $\bout_v(r'+R).$ By a variant of the standard ball-growing inequality (\cref{lemma:inbound}) we can show that a good cut always exists.

We now will give some intuition about the condition in \GoodCut~and the (somewhat strange) appearance of the $O(\log\log n)$ in our algorithm. First, we remark that the condition in \GoodCut~must track both the number of vertices and edges in the ball: the former to control recursion depth and roundtrip cover size, and the latter to control runtime.  Now we give intuition for why we require an $O(k \log\log n)$ approximation factor in our algorithm. Consider growing inballs $\bin_v(r)$ from $v$ for various radii $r$, and recall that we make a cut depending on the relative sizes of $|V(\bin_v(r))|$ and $|V(\bin_v(r + R))|$. Now, note that if for example $|V(\bin_v(r))| = O(1)$, we can afford to have $|V(\bin_v(r + R))| = O(n^{1/k})$, as we can simply run a naive algorithm on $\bin_v(r + R)$ now.  On the other hand, if for example $|V(\bin_v (r))| = \Omega(n)$, we can essentially only afford to have $|V(\bin_v(r + R))| \le \left(1 + \frac{1}{k}\right)|V(\bin_v(r))|.$ To see the latter, note that the recurrence $T(m) = \left(1 + \frac1k\right)(T(m/2) + T(m/2))$ has solution $T(m) = m^{1+O(1/k)}.$  Now, interpolating between these two extremes allows us to compute the optimal way to do ball cutting (which is done in \GoodCut). This leads to a ball cutting procedure with $O(k\log\log n)$ levels, and thus results in an $O(k \log \log n)$ approximation ratio.
\begin{algorithm}[h!] 
\caption{\RoundtripCover$(G, k, R)$, takes a graph $G$ with $n$ vertices, $m$ edges, and parameters $k$ and $R$. Returns a $(O(k\log\log n), R)$ roundtrip cover $C = \{B_1, B_2, \dots, \}$}
\begin{algorithmic}[1]
\State $\iin, \iout \assign 0$.
\State $r \assign 5kR\log\log n$.
\State Take any $v \in V(G).$
\While{\textbf{true}} \emph{\textbackslash\textbackslash some condition below in lines \ref{line:bothbig}, \ref{line:goodcutin}, \ref{line:goodcutout} will trigger eventually}
	\If{$\min(|V(\bin_v((\iin+1) R))|, |V(\bout_v( (\iout+1) R))| \ge \frac{3n}{4}$} \label{line:bothbig}
		\State \label{line:cuthalf} \Return $\{B_v( 2r+R) \} \cup$ \RoundtripCover$(G\backslash (\bin_v( (\iin+1) R) \cap \bout_v( (\iout+1) R)), R, k)$.
	\EndIf
	\If{\GoodCut$(G, \bin_v(\iin R), \bin_v((\iin+1)R))$} \label{line:goodcutin}
		\State \label{line:cutin} \Return \RoundtripCover$(\bin_v( (\iin+1) R), R, k) \cup $\RoundtripCover$(G\backslash \bin_v( \iin R), R, k)$.
	\EndIf
	\If{\GoodCut$(G, \bout_v(\iout R), \bout_v((\iout+1)R))$} \label{line:goodcutout}
		\State \label{line:cutout} \Return \RoundtripCover$(\bout_v( (\iout+1) R), R, k) \cup $\RoundtripCover$(G\backslash \bout_v( \iout R), R, k)$.
	\EndIf
	\If{$|E(\bin_v( \iin R))| \le |E(\bout_v( \iout R))|$ or $|V(\bout_v( \iout R))| \ge \frac{3n}{4}$} \label{line:growin}
		\State $\iin \assign \iin + 1$
	\Else
		\State $\iout \assign \iout + 1$
	\EndIf
\EndWhile
\end{algorithmic}
\label{algo:roundtripcover}
\end{algorithm}
\begin{algorithm}[h!] 
\caption{\GoodCut$(G, B_1, B_2)$, takes a graph $G$ with $n$ vertices and $m$ edges, balls $B_1 \subseteq B_2$, and determines whether recursing on $B_2$ and then deleting $B_1$ from our graph is good progress}
\begin{algorithmic}[1]
\If{$|V(B_2)| \le \frac{3}{4}n$ and $|V(B_2)| \le |V(B_1)|^\frac{k-1}{k} n^\frac{1}{k}$ and \\ $|E(B_2)| \le \max((1+\frac{1}{k})|E(B_1)|, |E(B_1)|^\frac{k-1}{k} m^\frac{1}{k})$}
	\State \Return \textbf{true}
\Else
	\State \Return \textbf{false}
\EndIf
\end{algorithmic}
\label{algo:goodcut}
\end{algorithm}
\paragraph{Explanation of \cref{algo:roundtripcover}:} We now explain what each piece of \cref{algo:roundtripcover} is doing. Here, $\iin$ and $\iout$ track the radius of the inball and outball that we are growing. We grow the balls at the same rate. If we notice that at any point we are in position to make a good cut (see lines \ref{line:goodcutin}, \ref{line:goodcutout}) then we do so. Otherwise, we know that both balls will eventually contain many vertices (see line \ref{line:bothbig}). In this case, we add $B_v(2r+R)$ to our roundtrip cover, delete $\bin_v(\iin R) \cap \bout_v(\iout R)$ from our graph, and recurse. To grow the inball and outball at the same rate, we run Dijkstra to grow the inball and outball, alternately processing an edge at a time from the inball and outball. We check the condition of \GoodCut~on a ball when we have certified that we have processed all vertices up to distance $\iin R$ or $\iout R$ respectively.
\subsection{Analysis of \RoundtripCover~and proof of \cref{thm:cover}}
\label{sec:analysis}
In this section we prove \cref{thm:cover}, bounding the performance of our roundtrip cover algorithm \cref{algo:roundtripcover}. We start by showing that $1+\max(\iin, \iout) \le 5k\log\log n$ at all points in the algorithm, hence some condition in lines \ref{line:bothbig}, \ref{line:goodcutin}, \ref{line:goodcutout} will trigger eventually.
\begin{lemma}
\label{lemma:inbound}
At all points during \cref{algo:roundtripcover}, we have that $1+\max(\iin, \iout) \le 5k\log\log n.$
\end{lemma}
\begin{proof}
We show $1 + \iin \le 5k\log\log n$, and the bound on $1 + i_{\iout}$ is analogous. To prove this we assume that none of the conditions in the inner loop of the algorithm trigger, and compute the resulting vertex and edge sizes of $\bin_v(\iin R)$ and $\bout_v(\iout R)$. To this end, assume that $|V(\bin_v( \iin R))| \le \frac{3n}{4}$ and $|E(\bin_v( \iin R))| \le m.$ By the conditions of lines \ref{line:bothbig}, \ref{line:goodcutin}, and \ref{line:growin} we know that each time we increment $\iin$ either
\begin{equation}\label{eq:vert} |V(\bin_v( (\iin+1)R))| \ge |V(\bin_v( \iin R))|^\frac{k-1}{k} n^\frac1k
\end{equation}
or
\begin{align}\label{eq:edge1}|E(\bin_v((\iin+1)R))| &\ge |E(\bin_v(\iin R))|^\frac{k-1}{k} m^\frac1k \text{ and }\\ \label{eq:edge2} |E(\bin_v( (\iin+1)R))| &\ge \left(1+\frac{1}{k}\right)|E(\bin_v( \iin R))|.
\end{align}
We first show that \cref{eq:vert} can only hold for $2k \log 4\log n$ values of $\iin$. To this end, define a sequence $\{x_i\}_{i \ge 0}$ as $x_0 = 1$ and $x_{i+1} = x_i^\frac{k-1}{k} n^\frac1k.$ By induction it follows that $x_i = n^{1 - \left(\frac{k-1}{k}\right)^i}.$ In particular, \[ x_{2k\log4\log n} = n^{1 - \left(\frac{k-1}{k}\right)^{2k\log 4\log n}} \ge \frac{3}{4}n. \] This shows that the condition in \cref{eq:vert} can only hold at most $2k \log 4\log n$ times. Similarly, after \cref{eq:edge1} holds for $2k \log 4\log n$ different $\iin$, we will have that $|E(\bin_v(\iin R))| \ge \frac{3m}{4}.$
At this point, \cref{eq:edge2} can hold at most $k$ times. This gives us that in total \[ 1+\iin \le 1 + 2k \log 4\log n + 2k \log 4\log n + k \le 5k \log\log n \] as desired.
\end{proof}
Now we proceed to proving \cref{thm:cover}.
\begin{proof}[Proof of \cref{thm:cover}]
We first show that the algorithm indeed returns a $(O(k\log\log n), R)$ roundtrip cover. Then we bound the total size of balls in the roundtrip cover, as well as the runtime.
\paragraph{Returns a $(O(k\log\log n), R)$ roundtrip cover.} We analyze lines \ref{line:cuthalf}, \ref{line:cutin}, and \ref{line:cutout}. In line \ref{line:cuthalf}, note that by \cref{lemma:inbound}, we know that $(\iin+1)R, (\iout+1)R \le r.$ Therefore, we know that $\bin_v( \iin R) \cap \bout_v( \iout R) \subseteq B_v( 2r).$ Additionally, it is clear that for any vertex $u \in \bin_v( \iin R) \cap \bout_v( \iout R)$, if another vertex $u'$ satisfies $d(u \lra u') \le R$ then $u' \in B_v( 2r+R).$ Therefore, the ball $B_v( 2r+R)$ contains both $u$ and $u'$, so we can safely delete $\bin_v( \iin R) \cap \bout_v( \iout R)$ from $G$ and recurse. This is exactly what is happening in line \ref{line:cuthalf}. In line \ref{line:cutin}, note that for any vertex $u \in \bin_v( \iin R)$, if another vertex $u'$ satisfies $d(u \lra u') \le R$ then $u' \in \bin_v( (\iin+1)R).$ Therefore, if we construct a roundtrip cover on $\bin_v( (\iin+1)R)$, then we can safely delete $\bin_v( \iin R)$ from $G$ and recurse. This is exactly what occurs in line \ref{line:cutin}. The same argument now applies to line \ref{line:cutout}. Finally, note that all balls we create are of radius $2r+R = O(Rk \log\log n).$
\paragraph{Total sizes of balls is $n^{1+O(1/k)}$.} We show by induction that the total number of vertices among all balls in the rountrip cover computed is at most $10n^\frac{k}{k-1}$ for an input graph $G$ with $n$ vertices. We show this by analyzing lines \ref{line:cuthalf}, \ref{line:cutin}, and \ref{line:cutout}. For line \ref{line:cuthalf}, note that because $\min(|V(\bin_v( (\iin+1) R))|, |V(\bout_v( (\iout+1) R))|) \ge \frac{3n}{4}$, we know that $|V(\bin_v( (\iin+1) R)) \cap V(\bout_v( (\iout+1) R))| \ge \frac{n}{2}$.
Therefore, it suffices to verify \[ 2n + 10\left(\frac{n}{2}\right)^\frac{k}{k-1} \le 10n^\frac{k}{k-1} \] which is clear. For line \ref{line:cutin}, for simplicity let $s = |V(\bin_v( \iin R))|.$ Then by the condition of \GoodCut, it suffices to note that \begin{align*} &10|V(\bin_v( (\iin+1)R))|^\frac{k}{k-1} + 10(n-s)^\frac{k}{k-1} \le 10(s^\frac{k-1}{k} n^\frac{1}{k})^\frac{k}{k-1} + 10(n-s)^\frac{k}{k-1} \\ &\le 10s n^\frac{1}{k-1} + 10(n-s) n^\frac{1}{k-1} = 10n^\frac{k}{k-1}. \end{align*} The same argument now applies to line \ref{line:cutout}.
\paragraph{Can be implemented to run in time $m^{1+O(1/k)}$.} We can implement the algorithm to grow $\bin_v( \iin R)$ and $\bout_v( \iout R)$ at the same rate, i.e., we process a single inedge and outedge at a time, and increment $\iin$ and $\iout$ when we are sure that we've processed the whole inball or outball. This can be done with Dijkstra's algorithm. We stop growing a ball once it contains at least $\frac{3n}{4}$ vertices. This way, any time we recurse, the total amount of work we have done to this point is at most twice the number of edges in the piece we are recursing on in lines \ref{line:cuthalf}, \ref{line:cutin}, and \ref{line:cutout}. To bound the runtime, we imagine lines \ref{line:cutin} and \ref{line:cutout} as partitioning the graph into pieces of the form $\bin_v( \iin R)$ or $\bout_v( \iout R)$ and then recursing on $\bin_v( (\iin+1)R)$ or $\bout_v( (\iout+1) R)$. This way, the depth of the recursion is at most $O(\log n)$ because we know that $|V(\bin_v( (\iin+1)R))|, |V(\bout_v( (\iout+1) R))| \le \frac{3n}{4}$ when we recurse.

We will now show that the total number of edges in level $\ell$ of the recursion is bounded by $\left(1+\frac{2}{k}\right)^\ell m^\frac{k}{k-1}$, where the top level is level $0$. We proceed by induction on $\ell.$ Say that the algorithm partitions $G$ into $G = G_1 \cup G_2 \cup \dots \cup G_j,$ where each $G_i$ is either of the form $\bin_v( \iin R)$ or $\bout_v( \iout R).$ For simplicity, let $s_i = |E(G_i)|$ and let $t_i = |E(\bin_v( (\iin+1) R))|$ or $t_i = |E(\bout_v( (\iout+1) R))|$ corresponding to what $G_i$ was. We know by the condition of \GoodCut~that $t_i \le \max(\left(1+\frac{1}{k}\right)s_i, s_i^\frac{k-1}{k} m^\frac{1}{k}).$ By induction, we know that the total number of edges processed in level $\ell$ is at most \begin{align*} &\sum_i \left(1+\frac{2}{k}\right)^{\ell-1} t_i^\frac{k}{k-1} \le \left(1+\frac{2}{k}\right)^{\ell-1} \sum_i \max\left(\left(1+\frac{1}{k}\right)s_i, s_i^\frac{k-1}{k} m^\frac{1}{k}\right)^\frac{k}{k-1} \\ &\le \left(1+\frac{2}{k}\right)^{\ell-1} \sum_i \left(1+\frac{2}{k}\right)s_i m^\frac{1}{k-1} \le \left(1+\frac{2}{k}\right)^\ell m^\frac{k}{k-1} \end{align*} as $\sum_i s_i \le m$ obviously.

Now, it is clear that the total work done on a graph $G$ at some node of the recursion tree is $\tilde{O}(|E(G)|)$ as line \ref{line:cuthalf} only occurs $O(\log n)$ times. Now taking $\ell = O(\log n)$ in the above claim completes the proof.
\end{proof}
\subsection{Proofs of \cref{thm:girth} and \cref{thm:spanner}}
\label{sec:mainappl}
Both theorems follow easily from \cref{thm:cover}.

\begin{proof}[Proof of \cref{thm:girth}]
We first show the result for unweighted graphs. To show this, run \[ \text{\RoundtripCover}(G, O(k), 2^i) \text{ for } 0 \le i \le O(\log n). \] Now, set our estimate $g'$ of the girth to be the smallest radius of any nontrivial ball that we had in a roundtrip cover. By the guarantees of \RoundtripCover, it is clear that $g \le g' \le O(k\log\log n) \cdot g$ as desired. It is clear that the algorithm runs in time $\tilde{O}(m^{1+\frac1k})$ by \cref{thm:cover}. We can extend this to weighted graphs by instead taking $0 \le i \le O(\log nW)$, where $W$ is the maximum edge weight. This can be improved to $O(\log n)$ by the same method as done in \cite{PRSTV18}, where they give a general reduction by contracting small weight strongly connected components and deleting large weight edges (see Section 5.1 in \cite{PRSTV18} for more details).
\end{proof}
\begin{proof}[Proof of \cref{thm:spanner}]
We first show the result for unweighted graphs. It is easy to see that \[ \bigcup_{i = 0}^{O(\log n)} \text{\RoundtripCover}(G, O(k), 2^i) \] is an $O(k \log\log n)$ spanner with $\tilde{O}(n^{1+ 1/k})$ edges by \cref{thm:cover}. It is clear that the algorithm runs in time $\tilde{O}(m^{1+\frac1k}).$ The extension to weighted graphs follows as in the above paragraph (proof of \cref{thm:girth}).
\end{proof}

\section{An $O(k \log k)$ Approximation in $\O(m^{1+1/k})$ Time}
\label{sec:klogk_approx}
In this section we explain how to combine the ideas from \cref{algo:roundtripcover} and \cref{algo:girth3} to give an algorithm for $(O(k \log k), R)$-roundtrip covers with $\O(n^{1+ 1/k})$ edges in time $\O(m^{1 + 1/k})$. Then \cref{thm:const_girth} and \cref{thm:arb_const_spanner} follow from this in the same way that \cref{thm:girth} and \cref{thm:spanner} followed from \cref{thm:cover}.

\begin{theorem}[Improved Randomized Roundtrip Cover]
\label{thm:roundtripcover2}
For an $n$-vertex $m$-edge graph $G$, an execution of $\RoundtripCoverNew(G, k, R)$ returns a collection $C$ of balls that form a $(O(k\log k), R)$ roundtrip cover of a weighted directed graph $G$ in time $m^{1+O(1/k)}$ where $\sum_{B \in C} |V(B)| = n^{1+O(1/k)}.$
\end{theorem}

The remainder of the section is organized as follows. We first give an overview for our approach, which combines the complementary approaches of sections \cref{sec:algo} and \cref{sec:constantapprox}. We then state our main algorithm, \cref{algo:roundtripcover2}. Afterwards, we analyze \cref{algo:roundtripcover2} to prove \cref{thm:roundtripcover2} in \cref{sec:analysis2}. Finally, we apply \cref{thm:roundtripcover2} to prove \cref{thm:const_girth} and \cref{thm:arb_const_spanner}.

\paragraph{Overview of approach.} Throughout this section, we assume that we have applied \cref{lemma:regular} to make our graph $G$ approximately regular. Here we give a high level overview for the ideas behind the algorithm. Let $G$ be an $n$-vertex $m$-edge graph and let $K \defeq 10k \log k$ for integer $k$. We start by generalizing \cref{algo:girth3} and \cref{alg:similar_vertices} slightly, where we consider the case where the sampled sets $S_i$ have size $\O(n^{1/k})$ instead of $\O(n^{1/2}).$ To elaborate, we first view \cref{algo:girth3} and \cref{alg:similar_vertices} as algorithms with the following guarantees. They add $\O(n^{3/2})$ edges towards a spanner, and then for each vertex $v$ which is not yet in a cycle of length $4R$ using the current spanner edges builds a data structure $D_v$ (corresponding to \cref{alg:similar_vertices}) which certifies that for all but at most $O(n^{1/2})$ other vertices $u$ we have that $d(v \lra u) > 4R$. We can generalize this as follows. There is a corresponding algorithm (\cref{algo:buildsimilar}) which has the following guarantees. It adds $\O(n^{1+1/k})$ edges towards a spanner, and then for each vertex $v$ which is not yet in a cycle of length $2KR$ using the current spanner edges builds a data structure $D_v$ which certifies that for all but at most $O(n^\frac{k-1}{k})$ other vertices $u$ we have that $d(v \lra u) > KR$.

After running this generalized algorithm (\cref{algo:buildsimilar}), for a vertex $v$, we can define \emph{$i$-similar vertices} to $v$, which are intuitively the vertices that the data structure $D_v$ thinks could still possibly be in a cycle of length $kR$ with $v$ and which are within distance $iR$ of $v$. Then we define a sequence $E_v^0, E_v^1, \cdots, E_v^K$ of ``balls" centered at $v$, where $E_v^i$ is the outball from $v$ consisting of $i$-similar vertices. The following important conditions hold: $v \in E_v^0,$ and $E_v^i \subseteq E_v^{i+1}$ for all $0 \le i < K.$ Finally, if $u \in E_v^i$ and $d(u \lra u') \le R$, then $u' \in E_v^{i+1}.$ This allows us to apply the ball-growing procedure in \cref{algo:roundtripcover} but using the balls $E_v^i$. Note that by our choice of $K$ and a variant of \cref{lemma:inbound}, there exists a good cut. This is because \[ n^{1 - \left(\frac{k-1}{k}\right)^K} \ge n^{1 - \frac{1}{k}} = n^\frac{k-1}{k}. \] Hence, we can make this good cut and then recurse. Here, our cutting condition is simpler (only checks vertices, not edges) because we have reduced to the case of regular graphs through \cref{lemma:regular}.

\begin{algorithm}[p]
\caption{$\RoundtripCoverNew(G, k, R)$. Takes in a $n$-vertex $m$-edge graph $G$, parameter $k$, and distance $R$. Returns a $(O(k\log k), R)$ roundtrip cover $C = \{B_1, B_2, \dots, \}$}
\begin{algorithmic}[1]
\State $C \assign \emptyset$.
\State $(G', C', D) \assign \BuildSimilar(G, k, R)$
\State $C \assign C'.$
\State $C \assign C \cup \BallGrow(G', k, R, D).$ \label{line:callballgrow}
\State \Return $C$.
\end{algorithmic}
\label{algo:roundtripcover2}
\end{algorithm}

\begin{algorithm}[p]
\caption{$\BuildSimilar(G, k, R)$. Takes in a $n$-vertex $m$-edge graph $G$, parameter $k$, and distance $R$. Returns a triple $(G', C, D)$, where $G' \subseteq G$ is a subgraph which still needs to be processed, $C$ is a set of balls to include in the roundtrip cover, and $D$ is a data structure which supports similarity queries. $\hn$ is the number of vertices at the top level of recursion.}
\begin{algorithmic}[1]
\State $C \assign \emptyset.$
\State $K \assign 10k \log k.$
\State Select uniformly random subsets $S^1, S^2, \cdots, S^{100\log \hn} \subseteq V(G)$ where $|S^i| = 100n^\frac1k \log^2 \hn \forall i$.
\State For all vertices $u \in S^i$ for some $i$, build a shortest path tree to and from $u$.
\State $C \assign \bigcup_{i=1}^{100\log \hn} \bigcup_{u \in S^i} B_u((4K+1)R).$ \label{line:updcover}
\For{$v \in V(G)$}
	\If{$v \in \left(\bin_u(2KR) \cap \bout_u(2KR)\right)$ for some $u \in S^i$ for some $i$}
		$\ON[v] \assign \textbf{false}.$ \label{line:markdead}
	\EndIf
\EndFor
\For{$v \in V(G)$} \label{line:searchv}
	\For{$i = 1$ to $K$}
		\State $T_v^i = \{ u \in S^i : d(v, u) \le KR \text{ and } d(u, w) \le 2KR \forall w \in S_v^j \forall 1 \le j < i. \}$
		\If{$|T_v^i| \ge 50\log \hn$} \label{line:maketi}
			\State $S_v^i \assign $ uniform sample of $T_v^i$ of size $50\log \hn$.
		\Else
			\State Return to line \ref{line:searchv}.
		\EndIf
	\EndFor
\EndFor
\State Have $D$ store all the $S_v^i$ and shortest path trees from all vertices $u \in S^i$ for some $i$.
\State \Return $\left(G\left[\{v : \ON[v] = \textbf{true}\}\right], C, D \right).$
\end{algorithmic}
\label{algo:buildsimilar}
\end{algorithm}

\begin{algorithm}[p]
\caption{$\BallGrow(G, k, R, D)$. Takes in a $n$-vertex $m$-edge graph $G$, parameter $k$, distance $R$, and data structure $D$ supporting similarity queries. Returns a set $C$ of balls to include in the roundtrip cover.}
\begin{algorithmic}[1]
\State $C \assign \emptyset$.
\State $K \assign 10k \log k$.
\State $\ON[v] \assign \textbf{true} \forall v \in V(G).$
\While{there exists $v$ with $\ON[v] = \textbf{true}$} \label{line:findv}
	\For{$i = 0$ to $K-1$}
		\State $E_v^i \assign \{ u \in V(G) : \ON[u] \text{ and } \SimilarTest(G, u, v, D, i, R) \text{ and } u \text{ reachable from } v \text{ through } E_v^i\}$. \label{line:makeevi} \Comment{We elaborate on this definition \cref{sec:explanation}.}
		\State \footnotesize$E_v^{i+1} \assign \{ u \in V(G) : \ON[u] \text{ and } \SimilarTest(G, u, v, D, i+1, R) \text{ and } u \text{ reachable from } v \text{ through } E_v^{i+1}\}$.
		\normalsize
		\If{\GoodCutNew$(G, E_v^i, E_v^{i+1})$} \label{line:goodcuttrue}
			\State $C \assign C \cup \RoundtripCoverNew(E_v^{i+1}, k, R)$.
			\State $\ON[v] \assign \textbf{false} \forall v \in E_v^i.$
			\State Break loop and return to line \ref{line:findv}.
		\EndIf
	\EndFor
\EndWhile
\State \Return $C$.
\end{algorithmic}
\label{algo:ballgrow}
\end{algorithm}

\begin{algorithm}[p]
\caption{\SimilarTest$(G, u, v, D, i, R)$, Takes in a $n$-vertex $m$-edge graph $G$, vertices $u, v \in V(G)$, data structure $D$, parameter $R$, decides whether $u$ is $i$-similar to $v$}
\begin{algorithmic}[1]
\State $K \assign 10k \log k.$
\If{$d(v, u) > iR$}
	\State \Return \textbf{false}
\EndIf
\For{$1 \le j \le 100\log \hn$}
	\For{$w \in S_v^j$}
		\If{$d(u, w) > (i+K)R$}
			\State \Return \textbf{false}
		\EndIf
	\EndFor
\EndFor
\Return \textbf{true}
\end{algorithmic}
\label{algo:similartest}
\end{algorithm}

\begin{algorithm}[p]
\caption{\GoodCutNew$(G, B_1, B_2)$, takes a graph $G$ with $n$ vertices and $m$ edges, balls $B_1 \subseteq B_2$, and determines whether recursing on $B_2$ and then deleting $B_1$ from our graph is good progress}

\begin{algorithmic}[1]
\If{$V(B_2) \le n^\frac1k |V(B_1)|^\frac{k-1}{k}$}
	\State \Return \textbf{true}
\Else
	\State \Return \textbf{false}
\EndIf
\end{algorithmic}
\label{algo:goodcut2}
\end{algorithm}
\subsection{Explanation of algorithms}
\label{sec:explanation}
\paragraph{Explanation of \cref{algo:roundtripcover2}, \cref{algo:buildsimilar}, \cref{algo:ballgrow}, \cref{algo:similartest}, \cref{algo:goodcut2}.} 
Throughout, we let $\hn$ be the number of vertices at the top level of recursion in the algorithms and we let $K \defeq 10k \log k$ for integer $k$.

We start by explaining \cref{algo:buildsimilar} (\BuildSimilar), which builds a data structure which allows efficient similarity queries. It follows the same blueprint as \cref{alg:similar_vertices}. The algorithm first selects sets $S^i$ for $1 \le i \le 100 \log \hn$, where $|S^i| = 100n^{1/k} \log^2 \hn$ for all $i$. It then computes shortest path trees to and from all vertices in all $S^i$. The algorithm then adds roundtrip balls of radius $O(KR)$ centered at each $u \in S^i$ to our roundtrip cover. The algorithm then marks all vertices $v$ from the graph that are within distance $2KR$ both to and from some vertex $u$ in some $S^i$ as not turned on anymore. Then for all vertices $v \in V(G)$ the algorithm builds sets $T_v^i$ and a uniform sample $S_v^i$ of $T_v^i$ of size $O(\log \hn)$ that allow us to ``test" whether another vertex $u$ is similar to $v$, i.e. could potentially be in a cycle of length $O(KR)$ with $v$. Eventually, $|T_v^i|$ gets small, and the algorithm stops processing vertex $v$. Finally, it returns the graph $G'$ of all still on vertices, the updated roundtrip cover, and the data structure $D$ for similarity testing consisting of all the $S_v^i$ for each vertex $v$ and shortest path trees from all $u \in S^i$.

Now we explain \cref{algo:similartest} (\SimilarTest), which uses the data structure $D$ computed by \BuildSimilar~to decide whether vertex $u$ is $i$-similar to vertex $v$. It returns true if and only if \[ d(v, u) \le iR \text{ and } d(u, w) \le (i+K)R \forall w \in S_v^j \forall 1 \le j \le 100\log \hn. \] Intuitively, this contains a ball around $v$ of distance $iR$ that contains all vertices which could potentially be in a cycle of length $KR$ with $v$, according to the algorithm.

Now we explain \cref{algo:goodcut2} (\GoodCutNew), which decides whether cutting out ball $B_1$ and recursing on $B_2$ constitutes good enough progress. This simply takes as input two balls $B_1$ and $B_2$ and decides whether recursing on $B_2$ and then deleting $B_1$ is good enough progress in trying to achieve a $\O(n^{1+O(k^{-1})})$ total size of roundtrip covers. Here, we check only the vertex condition instead of the edge condition (different from \cref{algo:goodcut} \GoodCut) because we have already reduced to the case where our graph $G$ is approximately regular (\cref{lemma:regular}).

Now we explain \cref{algo:ballgrow} (\BallGrow), which grows the balls $E_v^i$. $\ON[v] = \textbf{false}$ if vertex $v$ has been resolved, i.e. we can ensure that for any $u$ with $d(v \lra u) \le R$, that $v$ and $u$ are in a roundtrip ball of diameter $O(KR).$ Otherwise, $\ON[v] = \textbf{true}.$ We now grow balls $E_v^0, E_v^1, \cdots, E_v^K$ around $v$, up until line \ref{line:goodcuttrue} is satisfied. Our main claim is that when we recurse on $E_v^{i+1}$, then we can safely remove all vertices in $E_v^i$. While the definition in line \ref{line:makeevi} \[ E_v^i \assign \{ u \in V(G) : \ON[u] \text{ and } \SimilarTest(G, u, v, D, i, R) \text{ and } u \text{ reachable from } v \text{ through } E_v^i\} \] may seem recursive, all we mean is to say that we run a search from $v$, only keeping vertices which are $i$-similar, i.e. $\SimilarTest(G, u, v, D, i, R)$ is true.

Finally, our main algorithm \cref{algo:roundtripcover2} (\RoundtripCoverNew) first calls \BuildSimilar~to build the similarity data structure needed for \BallGrow. It also removes vertices from $G$ that were already resolved (i.e. in cycles of length $2KR$) to get a graph $G'$. Then it grows balls to partition $G'$ and recurse.
\subsection{Analysis}
\label{sec:analysis2}
In this section we analyze the above algorithms. We first show that the number of similar vertices to any vertex $v$ in $G'$ (line \ref{line:callballgrow}) is at most $n^\frac{k-1}{k}$ with high probability.
\begin{lemma}
\label{lemma:Ksizebound}
Consider an execution of $\RoundtripCoverNew(G_0, k, r)$ on an $\hn$-vertex vertex graph $G_0$.
Consider a recursive execution of $\RoundtripCoverNew(G, k, r)$ on an $n$-vertex $m$-edge graph $G$. Consider the resulting execution $\BallGrow(G', k, R, D)$ (line \ref{line:callballgrow}). With probability at least $1 - \hn^{-7}$ we have that for all $v \in V(G')$ that the number of vertices $u \in V(G')$ satisfying
\[ d(v, u) \le KR \text{ and } d(u, w) \le 2KR \forall w \in S_v^j \forall 1 \le j \le 100\log\hn \] is at most $n^\frac{k-1}{k}.$
\end{lemma}
\begin{proof}
We follow the same approach as the proofs in \cref{sec:constspanner}. Consider a vertex $v \in V(G)$. Define \[ H_v^i := \{ u \in V(G) : d(v, u) \le KR \text{ and } d(u, w) \le 2KR \forall w \in S_v^j \forall 1 \le j < i \}, \] i.e. all vertices $u \in V(G)$ which would ``pass" the $i$-th level similarity test for $v$. Our main claim is that if $|H_v^i| \ge n^\frac{k-1}{k}$, then we have that $\frac{|H_v^{i+1}|}{|H_v^i|} \le \frac{9}{10}$ with high probability. This implies the result, because if $|H_v^{100 \log \hn}| \ge n^\frac{k-1}{k}$ still, then we have that
\[ |H_v^{100 \log \hn}| \le \left(\frac{9}{10}\right)^{100\log\hn} n < 1, \] an obvious contradiction.

Now we show that if $|H_v^i| \ge n^\frac{k-1}{k}$, then we have that $\frac{|H_v^{i+1}|}{|H_v^i|} \le \frac{9}{10}$ with high probability. Note that by definition that $T_v^i = S^i \cap H_v^i$. It is direct to verify by a Chernoff bound that $|T_v^i| \ge 50 \log \hn$ with probability at least $1 - \hn^{-10}$ assuming that $|H_v^i| \ge n^\frac{k-1}{k}.$ By the definition of $S_v^i$ (a uniformly random subset of $T_v^i$ of size $50 \log \hn$) and symmetry we can think of $S_v^i$ simply as a uniformly random subset of $H_v^i$ of size $50 \log \hn.$

We now argue that for at least $\frac{9}{10}$ fraction of vertices in $w \in H_v^i$ we have that
\[ \Pr_{w' \in H_v^i}\left[d(w, w') \le 2KR\right] \le \frac{4}{5}, \] i.e. only $\frac{4}{5}$ fraction of vertices $w' \in H_v^i$ satisfy $d(w, w') \le 2KR$. Assume the contrary for contradiction. By the Pigeonhole principle, there are at least \[ \left(\frac{9}{10}\cdot \frac{4}{5} - \frac{1}{2}\right)|H_v^i|^2 = .22|H_v^i|^2 \] (unordered) pairs of vertices $w, w' \in H_v^i$ such that both $d(w, w') \le 2KR$ and $d(w', w) \le 2KR$. By the Pigeonhole principle again, there must be a vertex $w \in H_v^i$ for which at least $.44 |H_v^i|$ vertices $w'$ satisfy both $d(w, w') \le 2KR$ and $d(w', w) \le 2KR$, so $d(w, w') \le 4KR.$ Now, note that $.44 |H_v^i| \ge .44 n^\frac{k-1}{k}$ by our condition. We argue that this is impossible because $v$ should have been marked as not on with high probability in line \ref{line:markdead}. Indeed, the probability that $v$ failed to get marked as not on is at most
\[ \left(1 - \frac{.44n^\frac{k-1}{k}}{n}\right)^{\sum_{i=1}^{100\log \hn} |S^i|} \le 1 - \hn^{-20} \] as desired.

Now, consider the $\frac{9}{10}$ fraction of vertices $w \in H_v^i$ with \[ \Pr_{w' \in H_v^i}\left[d(w, w') \le 2KR\right] \le \frac{4}{5}. \] For each of these vertices, the probability that $d(w, w')$ for all $w' \in S_v^i$ is at most $\left(\frac{4}{5}\right)^{50\log\hn} \le 1-\hn^{-10}.$ By definition then, we have that $\frac{|H_v^{i+1}|}{|H_v^i|} \le \frac{1}{10}$ by definition, as the $\frac{9}{10}$ fraction of vertices in $H_v^i$ discussed in this paragraph will with high probability not be in $H_v^{i+1}.$
\end{proof}
We next claim that the ball growing scheme of \cref{algo:ballgrow} satisfies some important conditions, which intuitively make the $E_v^i$ look like balls of radius $iR$.
\begin{lemma}
\label{lemma:ballgrowworks}
Consider an execution of $\RoundtripCoverNew(G_0, k, R)$ on $n$-vertex $m$-edge graph $G_0$. Now, consider the resulting execution of $\BallGrow(G, k, R, D)$ on graph $G$. We have that in the execution for all $v \in V(G)$ that
\begin{enumerate}
\item $v \in E_v^0$.
\item $E_v^i \subseteq E_v^{i+1}$ for $0 \le i \le K-1$.
\item For $0 \le i \le K-1$, if $u \in E_v^i$ and $d(u \lra u') \le R$, then $u' \in E_v^{i+1}$.
\end{enumerate}
\end{lemma}
\begin{proof}
For the first claim, note that vertices $w \in S_v^j$ for all $j$ satisfy $d(v, w) \le KR$ by line \ref{line:maketi} of \BuildSimilar. Therefore, $v$ satisfies all conditions of being in $E_v^i$ on line \ref{line:makeevi} of \BallGrow.

The second claim is obvious from looking at the the definition of $E_v^i$ in line \ref{line:makeevi} of \BallGrow.

For the third claim, note that if $u \in E_v^i$ and $d(u \lra u') \le R$ then $d(v, u') \le d(v, u) + R = (i+1)R$. Also, for any $w \in S_v^j$ for $1 \le j \le 100\log \hn$ we have that $d(u', w) \le d(u, w) + R \le (i+1+K)R.$ Hence $u' \in E_v^{i+1}$ as desired.
\end{proof}
We now show the analogue to \cref{lemma:inbound}, specifically that for some iteration $0 \le i \le K-1$ in \BallGrow, we have that the condition in line \ref{line:goodcuttrue} triggers.
\begin{lemma}
Consider an execution of $\RoundtripCoverNew(G_0, k, R)$ on $n$-vertex $m$-edge graph $G_0$. Now, consider the resulting execution of $\BallGrow(G, k, R, D)$ on graph $G$. For some $0 \le i \le K-1$ we have that the condition in line \ref{line:goodcuttrue} is true, i.e. we get a good cut.
\end{lemma}
\begin{proof}
\label{lemma:inbound2}
The computation proceeds the same way as in \cref{lemma:inbound}. Assume for contradiction that the condition in line \ref{line:goodcuttrue} is never true, so \[ |V(E_v^{i+1})| > n^\frac1k |V(E_v^i)|^\frac{k-1}{k}. \] By \cref{lemma:ballgrowworks}, we know that $|V(E_v^0)| \ge 1$, as $v \in E_v^0.$ Therefore, one can check by induction that \[ |V(E_v^i)| \ge n^{1-\left(\frac{k-1}{k}\right)^i}. \] For $K = 10k\log k$ we have that
\[ |V(E_v^K)| \ge n^{1-\left(\frac{k-1}{k}\right)^K} > n^\frac{k-1}{k}, \] which contradicts \cref{lemma:Ksizebound}.
\end{proof}
We turn to proving \cref{thm:roundtripcover2}.
\begin{proof}
We break the analysis into pieces. We show that executing $\RoundtripCoverNew(G, k, R)$ on a $n$-vertex $m$-edge graph $G$ returns a $(O(k\log k), R)$ roundtrip cover $C$ with total size of all balls at most $n^{1+O\left(\frac1k\right)}$ in time $m^{1+O\left(\frac1k\right)}$ with high probability.
\paragraph{Returns a $(O(k\log k), R)$ roundtrip cover.} We first argue that in a call to $\BuildSimilar(G, k, R)$ that for all vertices $v$ where we marked $\ON[v] = \textbf{false}$ that $v$ is properly resolved, i.e. for any vertex $u$ with $d(v \lra u) \le R$ that $u$ and $v$ are in a roundtrip ball of radius at most $(4K+1)R.$ Indeed, note that if we mark $\ON[v] = \textbf{false}$, then there must have been a vertex $w$ for which $w \in S^j$ for some $j$, and $d(v \lra w) \le 2KR + 2KR = 4KR.$ Then $d(u \lra w) \le (4K+1)R$. We have added the ball $B_w((4K+1)R)$ to our roundtrip cover $C$, as desired (line \ref{line:updcover}).

The only other piece to verify is that when we mark $\ON[v] = \textbf{false}$ in an execution of $\BallGrow(G, k, R, D)$ that $v$ is properly resolved, i.e. that in some recursive subproblem we have that for all $u$ with $d(v \lra u) \le R$ that $v$ and $u$ are in a roundtrip ball of radius at most $O(KR).$ But this holds immediately by \cref{lemma:ballgrowworks}: if $v \in E_w^i$ for some $w$, and $d(v \lra u) \le R$, then $u \in E_w^{i+1}$ as desired.

\paragraph{Total size of balls in $C$ is $n^{1+O\left(\frac1k\right)}$ with high probability.} The analysis here follows closely to the corresponding paragraph in \cref{sec:algo}. We will show that the total size of all graphs processed in a single level of recursion the algorithm is at most $n^\frac{k}{k-1}$, where our initial call was $\RoundtripCoverNew(G, k, R)$ for a $n$-vertex $m$-edge graph $G$. Then, the total size of all graphs processed in the recursion is $\O(n^\frac{k}{k-1})$, as the recursion depth is at most logarithmic. Then the bound on total size of balls in $C$ follows as for a graph $G$ with $n$ vertices, the total size of balls added to $C$ during $\BuildSimilar(G, k, R)$ is at most
\[ \sum_{j=1}^{100\log n} |S^j| = \O(n^{1+\frac1k}). \]

To show that the total size of all graphs processed at recursion depth $\ell$ in the algorithm is at most $n^\frac{k}{k-1}$, we use induction. Indeed, this holds at the bottom level of recursion. Consider an execution of $\BallGrow(G, k, R, D)$ on an $n$-vertex $m$-edge graph $G$. Let $F_1^1, F_2^1, \cdots, F_t^1$ be all the balls $E_v^i$ for which line \ref{line:goodcuttrue} was satisfied (and we know that line \ref{line:goodcuttrue} is satisfied for some $i$ by \cref{lemma:inbound2}). Let $F_1^2, F_2^2, \cdots, F_t^2$ be the corresponding balls $E_v^{i+1}$. We have that $\sum_{i=1}^t |V(F_i^1)| \le n$, as we marked all vertices in $E_v^i$ as not on anymore if line \ref{line:goodcuttrue} was satisfied for $E_v^i$ and $E_v^{i+1}$. Additionally, by the condition of \GoodCutNew, we have that $|V(F_i^2)| \le n^\frac1k |V(F_i^1)|^\frac{k-1}{k}.$ By induction, the total sizes of all graphs processed at depth $\ell$ through recursion on $F_1^2, F_2^2, \cdots, F_t^2$ is at most
\[ \sum_{i=1}^t |V(F_i^2)|^\frac{k}{k-1} \le \sum_{i=1}^t \left(n^\frac1k |V(F_i^1)|^\frac{k-1}{k}\right)^\frac{k}{k-1} \le \sum_{i=1}^t n^\frac{1}{k-1} |V(F_i^1)| \le n^\frac{k}{k-1} \] as desired.

\paragraph{Can be implemented to run in time $m^{1+O\left(\frac1k\right)}$ with high probability.}
The analysis in the above section on the total size of balls, we know that the total number of vertices in all graphs processed during the algorithm $\RoundtripCoverNew(G, k, R)$ is at most $n^{1+O\left(\frac1k \right)}.$ As we have reduced to the case of regular graphs through \cref{lemma:regular}, the total number of edges in all graphs processed during $\RoundtripCoverNew(G, k, R)$ is at most $\O(\delta n^{1+O\left(\frac1k \right)}) \le \O(m^{1+O\left(\frac1k\right)})$ for $\delta = O(m/n).$

We now argue that the non-recursive runtime of $\RoundtripCoverNew(G, k, R)$ on a graph $G$ with $n$ vertices and $m$ edges is $\O(m^{1+1/k})$. We start by analyzing \cref{algo:buildsimilar} ($\BuildSimilar$). We have that in $\BuildSimilar(G, k, R)$ that $\sum_{i=1}^{100\log \hn} |S^i| \le \O(n^{1/k})$, where $\hn$ is the number of vertices in the graph at the top level of recursion. Therefore, building a shortest path tree to and from all vertices in $\bigcup_i S^i$ takes $\O(m^{1+1/k})$ time using Dijkstra's algorithm. Computing all the sets $T^i_v$ clearly takes time \[ n \cdot \sum_{i=1}^{100\log \hn} |S^i| \cdot K = \O(m^{1+1/k}).\] We proceed to analyze \cref{algo:similartest} (\SimilarTest). This clearly takes $\O(1)$ time per call, as we have precomputed all shortest path trees and distances to and from all vertices $u \in S^i$ in \cref{algo:buildsimilar}. Also, \cref{algo:goodcut2} (\GoodCutNew) also obviously takes $O(1)$ time per call.

Now we analyze \cref{algo:ballgrow} (\BallGrow). We can build the sets $E_v^i$ by running any search from $v$, only keeping vertices $u$ that satisfy $\SimilarTest(G, u, v, D, i, R).$ This takes time proportional to $\O(\delta |E_v^{i+1}|)$, where we have used that each call to \SimilarTest~takes $\O(1)$ time. In accounting for this runtime, we can push the contribution to the next recursion level (as we are recursing on $E_v^{i+1}$). Therefore, the total non-recursive runtime used is $\O(m^{1+1/k})$ as claimed for an input graph with $m$ edges.

Now, as the total number of edges over all graphs is $\O(m^{1+O\left(\frac1k\right)})$, the total runtime would also be $\O(m^{1+O\left(\frac1k\right)})$ as desired.
\end{proof}
We now use \cref{thm:roundtripcover2} to get multiplicative girth approximation and roundtrip spanners, proving \cref{thm:const_girth} and \cref{thm:arb_const_spanner}.
\begin{proof}[Proof of \cref{thm:const_girth}]
We first show the result for unweighted graphs. To show this, run \[ \text{\RoundtripCoverNew}(G, O(k), 2^i) \text{ for } 0 \le i \le O(\log n). \] Now, set our estimate $g'$ of the girth to be the smallest radius of any nontrivial ball that we had in a roundtrip cover. By the guarantees of \RoundtripCoverNew, it is clear that $g \le g' \le O(k\log k) \cdot g$ as desired. It is clear that the algorithm runs in time $\tilde{O}(m^{1+\frac1k})$ by \cref{thm:roundtripcover2}. We can extend this to weighted graphs by instead taking $0 \le i \le O(\log nW)$, where $W$ is the maximum edge weight. This can be improved to $O(\log n)$ by the same method as done in \cite{PRSTV18}, where they give a general reduction by contracting small weight strongly connected components and deleting large weight edges (see Section 5.1 in \cite{PRSTV18} for more details).
\end{proof}
\begin{proof}[Proof of \cref{thm:arb_const_spanner}]
We first show the result for unweighted graphs. It is easy to see that \[ \bigcup_{i = 0}^{O(\log n)} \text{\RoundtripCoverNew}(G, O(k), 2^i) \] is an $O(k\log k)$ spanner with $\tilde{O}(n^{1+\frac{1}{k}})$ edges by \cref{thm:roundtripcover2}. It is clear that the algorithm runs in time $\tilde{O}(m^{1+\frac1k}).$ The extension to weighted graphs follows as in the above paragraph (proof of \cref{thm:const_girth}).
\end{proof}

\section{Conclusion and Open Problems}
\label{sec:open}

In this paper we provided multiple results on computing round-trip spanners and multiplicative approximations to the girth of an arbitrary directed graph. Our results all either improve running times, decrease the use of randomness, or improve the approximation quality of previous results. Ultimately, this work brings the state-of-the art performance of  roundtrip spanners algorithms on directed graphs closer to matching that for undirected graphs.  

An immediate open problem left open by our work is to fully close the gap between algorithmic guarantees for spanners of undirected graphs and roundtrip spanners of directed graphs and provide a deterministic algorithm which for all $k$ in $\tilde{O}(m n^{1/k})$ time computes a $O(k)$ roundtrip spanner with $\tilde{O}(n^{1 + 1/k})$ edges.
This paper resolves this problem for $k = \Omega(\log n)$ and makes progress on it for smaller values of $k$; it is still open to resolve it for all $k$. 

Another key open problem is to further clarify the complexity of approximating the girth of a directed graph. Currently the only algorithms which provably outperform APSP for approximating the girth of a graph are Pachocki et. al. \cite{PRSTV18} and this paper. Consequently, all known girth approximation algorithms for directed graphs leverage techniques immediately applicable for spanner computation (with the sole possible exception of the algorithms of Section~\ref{sec:constantapprox}). Therefore, beyond improving roundtrip spanner routines to obtain an algorithm which can compute an $O(k)$-multiplicative approximation to the girth in $\tilde{O}(m n^{1/k})$, this suggests the even more challenging open problem of circumventing this ``spanner barrier'' to 
obtaining even faster running times. 
For undirected graphs, it possible to overcome this barrier in certain cases \cite{IR77, LL09, RW12, DKMS17}. However, some of the techniques used in these results are known not to extend to directed graphs, see e.g. \cite{RW12, PRSTV18}. Consequently, further clarifying the complexity of girth approximation beyond the spanner barrier with either improved algorithms or new conditional lower bounds remains an difficult and interesting frontier.

One final open problem is to improve the parallel complexity of these routines. Previous work on the efficient construction of roundtrip spanners \cite{PRSTV18} provided such a result. Further, there have been recent advances in the efficient parallel computation of reachability in directed graphs \cite{Fineman18,JLS19arxiv} and commute times of random walks. The combination of the ideas from these works with the results of this paper could be useful for  obtaining further improvements for the efficient parallel computation of girth and roundtrip spanners \cite{CohenKPPRSV17}.

\section{Acknowledgements}
\label{sec:acknowledgements}

We thank the anonymous reviewers for helpful feedback and suggestion of many of the open problems discussed in \Cref{sec:open}. We thank Jakub Pachocki, Liam Roditty, Roei Tov, and Virginia Vassilevska Williams for helpful discussions.

{\small
\bibliographystyle{alpha}
\bibliography{refs}}

\begin{appendix}
\section{Deferred proofs}
\subsection{Proof of \cref{lemma:regular}}
\label{proofs:regular}
\begin{proof}
The reduction is as follows.
Let $\delta= \ceil{|E|/|V|}$.
Replace all the outgoing edges from $v$ by a balanced $\delta$-tree with all weights of internal edges 0 (a balanced tree with degree $\delta$ where $v$ is the root of the tree and all edges of the tree are directed from the root) where each leaf in this tree is ``responsible''  for
$d$ of the outgoing edges of $v$, that is, each leaf has $d$ outgoing edges to $\delta$ (different) neighbors of $v$.
The weight of these edges are the original corresponding weight of the edges of $v$.
We set the weight of all edges in the balanced $\delta$-tree to be 0.
A similar process is done for the incoming edges of $v$ for every node $v\in V$.
It is not hard to verify that the number of new nodes created is proportional to the number of edges divided by $\delta$, that is, the number of new nodes is $O(m/\delta)= O(n)$. In addition, every two original nodes $u$ and $v$ that have a directed path in $G$, also have a directed path in the modified graph.
It is not hard to verify that all round trip distances in $G$ and $H$ are the same (this implies also that the girth of $H$ and $G$ is the same).
Moreover, given a cycle $C$ in $H$ one can easily find a cycle of the same length in $G$ by simply contracting the
$\delta$-trees of each vertex.
Finally, given a subgraph $H'$ of $H$, one can obtain a subgraph $G'$ of $G$ by simply contracting the
$\delta$-trees of each vertex. It is not hard to verify that all roundtrip distances (for pairs of vertices in $G$) in $H'$ and $G'$ are the same.
\end{proof}

\newcommand{\spannerapprox}{\textsc{SpannerApprox}}

\section{Constant Approximation Roundtrip Spanner in $\tilde{O}(m \sqrt{n})$ time}
\label{sec:constspanner}

In this section we provide a procedure that given a directed weighted graph $G$ and a target distance $R$, returns in $\tilde{O}(m \sqrt{n})$ time
a subgraph $H$ with $\tilde{O}(n^{3/2})$ expected number of edges such that $d_H(u \lra v) \leq 8R$ for every two vertices $u$ and $v$ such that $d_G(u \lra v) \leq R$.

The pseudocode of the algorithm is given in Figure \ref{algo:spanner3}.

\begin{algorithm}[h!]
	\caption{$\spannerapprox(G,R)$,  takes a graph $G$ and a parameter $R$.  Computes a roundtrip spanner for target roundtrip distance $R$}
	\label{algo:spanner3}
	\begin{algorithmic}[1]
        \State $H \gets (V,\emptyset)$
		\State Invoke $\similarsetspanner(G,R)$ to get a subset of edges $H^\outset$ and sets $A_v^\outset \subseteq V$ for each $v \in V(G)$. \label{line:similar_search_spanner_1}
		\State Invoke $\similarsetspanner(\Grev,R)$ where $\Grev$ is the graph where the direction of every edge is reversed to
                get a subset of edges $H^\inset$ and sets $A_v^\inset \subseteq V$ for each $v \in V(G)$.
        \State For each $v \in V(G)$ perform Dijsktra from $v$ in the graph induced by $A_v^{\outset} \cup A_v^{\inset}$
            and add the edges of this shortest paths tree to $H$.
		\label{line:search_union_spanner}
        \State Add the edges of $H^\outset$ and the reversed edges of $H^\inset$ to $H$.
		\State Return $H$. \label{line:final_constant_spanner}
	\end{algorithmic}
\end{algorithm}
\begin{algorithm}[h!]
\caption{$\similarsetspanner(G, R)$, takes a graph $G$ and a parameter $R$. This algorithm returns a subset of edges $H'$ (that will be added to the final spanner) and  a set of $A_v \subseteq V(G)$ of ``similar''
	vertices to $v$ (with respect to balls of radius $R/2$) for each $v \in V(G)$.}
\label{alg:similar_vertices_spanner}
\begin{algorithmic}[1]
\State  $H' \gets (V,\emptyset)$
\State  Let $U \subseteq V$ be a set of $O(100 \log{n} \cdot n^\frac12)$ expected size by sampling every vertex
$v \in V$  independently with probability $p' = \frac{100 \log{n}}{n^\frac12}$.
\State
Run Dijkstra to/from each vertex $v\in U$ and add to $H'$ the edges of all these shortest paths trees.
\State
Let $Z\subseteq V$ be the set of vertices that their roundtrip distance from every vertex in $U$ is more than $3R$,
that is, $Z = \{v\in V \mid d_G(v \lra u) > 3R \text{ for all } u \in U\}$.
\label{line:deleteSatisfiedVertices}
\State
For $M = 50\log n$, sample sets $S_0, S_1, \cdots, S_M \subseteq Z$, each of expected size $O(n^\frac12)$ by sampling every vertex
$v \in Z$  independently with probability $p = \frac{1}{n^\frac12}$.
\State  Run Dijkstra to/from each vertex $w\in S_i$ for every $1 \leq i \leq M$.

\For{every vertex $v \in Z$}
    \State Set $T_0(v) \gets \{s \in S_0 \mid d(v, s) \le R/2 \}$.
    \For{$i = 1 , ... , M$}
           \State  Let $R_{i-1}(v)$ be a sampled set of $100 \log{n}$ vertices from $T_{i-1}(v)$ (or $T_{i-1}(v)$ if $|T_{i-1}(v)| \leq 100 \log{n}$).
            \State
             $\text{Set }
                T_i(v) \gets \{ s \in S_{i} \mid d(v, s) \le R/2  \text{ and }
                d(s, t) \le 3R/2 \text{ for all } t\in \cup_{j\in[0,...,i-1]}{R_{j}(v)}\}.$
    \EndFor
    \State  Grow a shortest path tree $T(v)$ from $v$, only keeping vertices $s \in Z$ satisfying \[ d(v, s) \le R/2 \text{ and }
            d(s, t) \le 3R/2 \text{ for all } t\in \cup_{j\in[0,...,M]}{R_{j}(v)}.\]
            \State   Set $A_v$ to be the set of vertices in $T(v)$.
        \label{line:GrowSmallBallSpanner}
	\EndFor
	\Return $H'$ and $A_v$ for all $v \in V$
\end{algorithmic}
\end{algorithm}

\subsection{Analysis}
We next analyze the running time and correctness of our algorithm.

The sets $G_i(v)$ are defined similarly as in the previous section but restricted on vertices in $Z$.
That is, for a vertex $v$, define $G_0(v) = \{s \in Z \mid d(v, s) \le R/2 \}$ and
$G_i(v) = \{s \in Z \mid d(v, s) \le R/2 \text{ and }
d(s, t) \le 3R/2 \text{ for all } t\in \cup_{j\in[0,...,i-1]}{R_{j}(v)}\}$.

We start with bounding the stretch of the returned spanner.

\begin{lemma}
\label{lem:correctnessSpanner}
For every  two vertices $u$ and $v$ such that $d_G(u \lra v) \leq R$, $d_H(u \lra v) \leq 8R$.
\end{lemma}
\begin{proof}
As $d_G(u \lra v) \leq R$ there must be a shortest cycle $C$ of length at most $R$
that contains both $u$ and $v$.
If all the vertices on $C$  belong to $Z$ then using similar analysis
 to the proof in Lemma \ref{lem:correctness} one can show that $C \subseteq A_v^\outset \cup A_v^\inset$.
And as shortest path tree from $v$ in the induced graph $\subseteq A_v^\outset \cup A_v^\inset$ is added to $H$, it follows that
$d_H(v,u) \leq d_C(v,u)$. Similarly, we can show that $d_H(u,v) \leq d_C(u,v)$.
We now get that $d_H(v \lra u) \leq d_C(v,u)+ d_C(u,v) \leq R$.

Now assume there exists a vertex $x$ on $C$ such that $x \notin Z$.
As $x \notin Z$ then by construction there is a vertex $y \in U$ such that $d_G(y \lra x) \leq 3R$.

As the algorithm adds a shortest path tree from and to $y$ then the roundtrip distances from $y$ to all vertices in $V$ are the same in $H$ as in $G$.

We get that $d_H(y \lra v) = d_G(y \lra v) \leq d_G(y \lra x) + d_G(x \lra v) \leq 3R+R=4R$.
Similarly, we can show that $d_H(y \lra u) \leq 4R$.
It follows that $d_H(u \lra v) \leq d_H(y \lra v) + d_H(y \lra u) \leq 8R$, as required.
\end{proof}

We next turn to the analysis of the running time and the number of edges in the constructed spanner $H$.

We say that a vertex $v$ is $R$-cycle-rich if there are
at least $50 \sqrt{n} \log{n}$ other vertices $z$ such that $v$ and $z$ are on a cycle of length at most $3R$, namely,
$|\{z\in V \mid d_G(v \lra x) \leq 3R \}| \geq 50\sqrt{n}\log{n}$.

The next lemma shows that with high probability $Z$ does not contain any
$R$-cycle-rich vertices.

\begin{lemma}
\label{lem:noRichVertices}
With probability at least $1 - 1/n^9$, $Z$ does not contain
$R$-cycle-rich vertices.
\end{lemma}
\begin{proof}

Consider an $R$-cycle-rich vertex $v$.
Note that by construction if $U$ contains a vertex from the set $\{x\in V \mid d_G(v \lra x) \leq 3R \}$ then
$v\notin Z$.
By definition we have
$|\{x\in V \mid d_G(v \lra x) \leq 3R \}| \geq 50 \log{n} \sqrt{n}$.
For a fixed vertex $x \in \{x\in V \mid d_G(v \lra x) \leq 3R \}$ the probability that $x \notin U$ is
$(1-p')$ for $p' = \frac{100 \log{n}}{n^\frac12}$.
The probability that none of the vertices in $\{x\in V \mid d_G(v \lra x) \leq 3R \}$ belongs to $U$
is  $(1-p')^{|Y|} \leq (1-\frac{100 \log{n}}{n^\frac12})^{50 \log{n} \sqrt{n}} < 1/n^{10}$.

By union bound on all $R$-cycle-rich vertices, the probability that $Z$ contains
an $R$-cycle-rich vertex is at most $1/n^{9}$
\end{proof}

The next lemma essentially shows that w.h.p.
$|G_{i+1}(v)| \leq 0.8 |G_{i}(v)|$ for every $v\in V$ and $i \in [M]$ such that
$|G_{i}(v)| \geq 200 \sqrt{n} \log{n}$.

\begin{lemma}
\label{lem:ManyEdgesSpanner}
Assume $Z$ does not contain $R$-cycle-rich vertices and that for every $v$ and every $i$
 all vertices in $G_{i+1}(v)$  are $(v,i)$-dense.
Then, for every vertex $v$, index  $i \in [M]$ such that $|G_{i}(v)| \geq 200 \sqrt{n} \log{n}$
we have $|G_{i+1}(v)| \leq 0.8 |G_{i}(v)|$.
\end{lemma}
\begin{proof}
Assume, towards contradiction, that there exists a vertex $v$ and index $i$ such that
$|G_{i+1}(v)| \geq 0.8 |G_{i}(v)|$.

Similarly to the proof of Lemma \ref{lem:ManyEdges}, let $H'$
be  the directed graph $H'$ whose set of vertices is $G_{i}(v)$ and set of edges is the following.
For every vertex $u$ in $G_{i}(v)$ add an outgoing edge for every vertex $s$ such that $d(u,s) \leq 3R/2$.
And as was shown in Lemma \ref{lem:ManyEdges} the number of edges in $E(H')$ is at least $0.71|G_{i}(v)|^2$, that is, $|E(H')| \geq 0.71|G_{i}(v)|^2$.

We again give a credit of half for the endpoint of every edge in $H'$.
As there are no $R$-cycle-rich vertices we claim that the maximal credit a vertex $z$ can get is less than
$|G_{i}(v)|/2 + 50 \log{n} \sqrt{n}/2 \leq |G_{i}(v)|/2 + |G_{i}(v)|/8 = 5|G_{i}(v)|/8$.
To see this, note there are less than $50 \log{n} \sqrt{n}$ vertices in $G_{i}(v)$ for which
$z$ has both incoming and outgoing edge, so at most $50 \log{n} \sqrt{n}$ credit for these vertices.
For the rest of the vertices $y$ in $G_{i}(v)$ the maximum credit for $z$ is $1/2$ (as there could be at most one incident edge between $y$ and $z$ ).
Overall, we get that the total credit of $z$ is less than $|G_{i}(v)|/2 + 50 \log{n} \sqrt{n}/2 \leq 5|G_{i}(v)|/8$.

It follows that the total credit for all vertices, which is equal to the number of edges in $H'$, is less than
$5|G_{i}(v)|^2/8 < 0.71|G_{i}(v)|^2$, contradiction.
\end{proof}

The next lemma shows that the shortest path tree $T(v)$ for a vertex $v\in Z$ constructed in
Line \ref{line:GrowSmallBallSpanner} of the algorithm contains w.h.p. at most
$200 \sqrt{n} \log{n}$. This lemma will be used both to prove our running time and the total number of edges in the constructed spanner $H$.

\begin{lemma}
\label{lem:smallBallsSpanner}
With probability at least $1-1/n^8$ the shortest paths tree $T(v)$ contains at most
$200 \sqrt{n} \log{n}$ vertices for every $v \in V$.
\end{lemma}
\begin{proof}

By Lemma \ref{lem:noRichVertices} with probability at least $1/1-n^9$, $Z$ does not contain
$R$-cycle-rich vertices.

By union bound on all vertices $v\in V$
on Lemma ~\ref{lem:ManyEdges} (that also applies here),
with probability at least $1-2/n^{9}$,
all vertices in $G_{i+1}(v)$  are $(v,i)$-dense.

That is, with probability at least $1-(1/n^9 +2/n^9) > 1-1/n^8$ both
1. $Z$ does not contain $R$-cycle-rich vertices and
2. all vertices $u\in G_{i+1}(v)$ are $(v,i)$-dense for every $v \in Z$ and $i \in [M]$.

Assume this indeed the case.

By Lemma \ref{lem:ManyEdgesSpanner} for every vertex $v$, index  $i \in [M]$ such that $|G_{i}(v)| \geq 200 \sqrt{n} \log{n}$
we have $|G_{i+1}(v)| \leq 0.8 |G_{i}(v)|$.

Consider such a vertex $v \in Z$.
As $G_{i+1}(v) \subseteq G_{i}(v)$ and $|G_{i+1}(v)| \leq 0.8|G_{i}(v)|$ for every $i \in [M]$
then this implies that there exists an index $M'$ in $[M]$ such that
$|G_{M'}(v)| < 200 \sqrt{n} \log{n}$.
As the shortest path tree $T(v)$ contains only vertices from $G_{M'}(v)$ the lemma follows.
\end{proof}

\begin{lemma}
\label{lem:RunningTimeSpanner}
The expected running time of Algorithm $\spannerapprox$ is $O(m\sqrt{n}\log{n}+n\sqrt{n}\log^3{n}) = \tilde{O}(m \sqrt{n})$.
\end{lemma}
\begin{proof}

Procedure $\spannerapprox$ invokes twice Procedure $\similarsetspanner$ on the graph $G$ and on the reversed graph of $G$.

Consider one call for Procedure $\spannerapprox$.
The expected size of the set $U$ is $O(\log{n}\sqrt{n})$.
For each vertex $v$ in $U$ the algorithm computes Dijkstra in $O(m+n\log{n})$ time.
Therefore the expected time for computing Dijkstra from all vertices in $U$ is $O(m\sqrt{n}\log{n}+n\sqrt{n}\log^2{n})$.

As in the analysis of Lemma \ref{lem:totaltime} computing
Dijkstra from all vertices in $\cup_{i\in M}{S_i}$ takes additional
$O(m\sqrt{n}\log{n}+n\sqrt{n}\log^2{n})$ time.

By Lemma \ref{lem:smallBallsSpanner} with probability at least $1-1/n^8$ the shortest paths tree $T(v)$ contains at most
$200 \sqrt{n} \log{n}$ vertices for every $v \in Z$.
Hence, with probability $1-1/n^8$ computing all trees take
$O(n \sqrt{n}\delta\log{n}) = O(m \sqrt{n} \log{n})$.

With probability at most $1/n^8$ there exists a tree $T(v)$ that contains more
$200 \sqrt{n} \log{n}$, in any case the total computation in this case is bounded by $O(mn)$.
Since this happens with small probability it does effect the asymptotic bound of the expected running time.

The lemma follows.
\end{proof}

Finally, we conclude that the number of edges in the constructed spanner $H$ is
$\tilde{O}(n^{3/2})$

\begin{lemma}
\label{lem:NumberOfEdgesSpanner}
The expected number of edges in the constructed spanner $H$ is $O(n^{3/2}\log{n} )$.
\end{lemma}
\begin{proof}

First note that Procedure $\spannerapprox$ adds the set of edges $H^\outset$ and $H^\inset$
returned from the calls to Procedure $\similarsetspanner$ to the spanner $H$.

We show that $H^\outset$ (similarly $H^\inset$) contains in expectation $O(n^{3/2}\log{n} )$ number of edges.
The expected size of the set $U$ is $O(\sqrt{n}\log{n} )$ for each vertex $v$ in
$U$ the algorithm adds $O(n)$ edges (for the shortest path from $v$ and to $v$).
Thus, $O(n^{3/2}\log{n} )$ edges are added to $H^\outset$ for all vertices in $U$.

Procedure $\spannerapprox$ adds for every vertex $v \in Z$ a shortest paths tree in the induced graph of $A_v^{\outset} \cup A_v^{\inset}$.
By Lemma \ref{lem:smallBallsSpanner} with probability at least $1-1/n^8$ the set $A_v^{\outset} \cup A_v^{\inset}$ contains at most
$O(\sqrt{n} \log{n})$ vertices for every $v \in V$.
Hence, with probability at least $1-1/n^8$
the number of edges in each tree $T(v)$ is at most
$O(\sqrt{n} \log{n})$ for every $v \in Z$.
Therefore, at most $O(n^{3/2} \log{n})$ number of edges added to $H$ in total.

With probability at most $1/n^8$ there is a vertex $v$ such that
$A_v^{\outset} \cup A_v^{\inset}$ contains more than $O(\sqrt{n} \log{n})$ vertices.
In any case the number of edges added to $H$ in this case is bounded by $m$ and since this happens with very small probability it does not affect the asymptotic bound of the expected size of the spanner.
The lemma follows.
\end{proof}

We conclude this section with a proof of  \cref{thm:const_spanner}.

\begin{proof}[Proof of  \cref{thm:const_spanner}]
The algorithm invokes procedure $\spannerapprox(G,R)$ for every $R=(1+\epsilon)^i$ and takes as a spanner the union of the returned spanners from all these calls.

Now, the theorem easily follows from Lemmas \ref{lem:correctnessSpanner}, \ref{lem:RunningTimeSpanner} and \ref{lem:NumberOfEdgesSpanner}.
 \end{proof}
\end{appendix}

\end{document}